\documentclass[11pt]{article}

\usepackage[top=.9in, left=.9in, right=.9in,bottom=.9in]{geometry}

\usepackage{latexsym}
\usepackage[parfill]{parskip}
\usepackage{amssymb,amsmath, bm, bbm}
\usepackage{graphicx}
\usepackage{marvosym}
\usepackage{lscape}
\usepackage{rotating}
\usepackage{setspace}
\usepackage{threeparttable}
\usepackage{booktabs}
\usepackage{multirow}
\usepackage{enumitem}

\usepackage{caption}
\usepackage{subcaption}

\usepackage{xr}
\makeatletter
\newcommand*{\addFileDependency}[1]{
\typeout{(#1)}
\@addtofilelist{#1}
\IfFileExists{#1}{}{\typeout{No file #1.}}
}\makeatother

\newcommand*{\myexternaldocument}[1]{%
\externaldocument{#1}%
\addFileDependency{#1.tex}%
\addFileDependency{#1.aux}%
}

\myexternaldocument{supp} 

\usepackage{natbib}

\usepackage{color}
\usepackage[pdftex, bookmarksopen=true, bookmarksnumbered=true,
pdfstartview=FitH, breaklinks=true, urlbordercolor={0 1 0}, citebordercolor={0 0 1}, colorlinks,allcolors=blue]{hyperref}

\usepackage{dcolumn}
\newcolumntype{.}{D{.}{.}{-1}}
\newcolumntype{d}[1]{D{.}{.}{#1}}

\usepackage{amsthm}
\theoremstyle{definition}
\newtheorem{theorem}{Theorem}[section]

\newtheorem{proposition}[theorem]{Proposition}

\newtheorem{assumption}{Assumption}

\newcommand{\R}{\ensuremath{\mathbb{R}}}

\newcommand{\bbone}{\ensuremath{\mathbbm{1}}}

\newcommand{\E}{\ensuremath{\mathbb{E}}}

\newcommand{\Var}{\text{Var}}

\def\b1{\boldsymbol{1}}

\usepackage{xcolor}

\def\spacingset#1{\renewcommand{\baselinestretch}%
{#1}\small\normalsize} \spacingset{1}

\begin{document}

\pagestyle{plain}

\newcommand{\blind}{0}

\newcommand{\tit}{Estimating Racial Disparities in Emergency General Surgery} 

\if0\blind

{\title{\tit\thanks{
 The Pennsylvania Health Cost Containment Council (PHC4) is an independent state agency responsible for addressing the problems of escalating health costs, ensuring the quality of health care, and increasing access to health care for all citizens. While PHC4 has provided data for this study, PHC4 specifically disclaims responsibility for any analyses, interpretations or conclusions. Some of the data used to produce this publication was purchased from or provided by the New York State Department of Health (NYSDOH) Statewide Planning and Research Cooperative System (SPARCS). However, the conclusions derived, and views expressed herein are those of the author(s) and do not reflect the conclusions or views of NYSDOH. NYSDOH, its employees, officers, and agents make no representation, warranty or guarantee as to the accuracy, completeness, currency, or suitability of the information provided here. This publication was derived, in part, from a limited data set supplied by the Florida Agency for Health Care Administration (AHCA) which specifically disclaims responsibility for any analysis, interpretations, or conclusions that may be created as a result of the limited data set. The authors declare no conflicts. Rachel Kelz is funded by a grant from the National Institute on Aging, R01AG060612.}}
\author{ Eli Ben-Michael\thanks{Carnegie Mellon University, Pittsburgh, PA, Email: ebenmichael@cmu.edu}
\and Avi Feller\thanks{University of California, Berkeley, Berkeley, CA, Email: afeller@berkeley.edu}
\and Rachel Kelz\thanks{University of Pennsylvania, Philadelphia, PA, Email: rachel.kelz@pennmedicine.upenn.edu}
\and Luke Keele\thanks{University of Pennsylvania, Philadelphia, PA, Email: luke.keele@gmail.com}
}

\date{\today}

\maketitle
}\fi

\if1\blind
\title{\bf \tit}
\maketitle
\fi

\begin{abstract}

Research documents that Black patients experience worse general surgery outcomes than white patients in the United States. 
In this paper, we focus on an important but less-examined category: the surgical treatment of emergency general surgery (EGS) conditions, which refers to medical emergencies where the injury is internal, such as a burst appendix. Our goal is to assess racial disparities for common outcomes after EGS treatment using an administrative database of hospital claims in New York, Florida, and Pennsylvania, and to understand the extent to which differences are attributable to patient-level risk factors versus hospital-level factors, as well as to the decision to operate on EGS patients. To do so, we use a class of linear weighting estimators that re-weight white patients to have a similar distribution of baseline characteristics to Black patients. This framework nests many common approaches, including matching and linear regression, but offers important advantages over these methods in terms of controlling imbalance between groups, minimizing extrapolation, and reducing computation time. Applying this approach to the claims data, we find that disparities estimates that adjust for the admitting hospital are substantially smaller than estimates that adjust for patient baseline characteristics only, suggesting that hospital-specific factors are important drivers of racial disparities in EGS outcomes. We also find little evidence that the decision to operate exacerbates racial disparities.
\end{abstract}

\begin{center}
\end{center}

\clearpage
\spacingset{1.5} 


\singlespacing

\section{Introduction: Racial Disparities in Health Care}

A substantial research literature documents that Black patients experience worse general surgery outcomes than white patients in the United States 
\citep{alavi2012racial,esnaola2008race,cooper1996surgery,silber2015examining}. 
In this paper, we focus on racial disparities for a related but less-examined category: the surgical treatment of emergency general surgery (EGS) conditions. EGS refers to medical emergencies where (unlike for trauma) the injury is likely ``endogenous,'' such as a burst appendix \citep{shafi2013emergency}. More than 800,000 emergency operations in the United States each year are for EGS conditions, and such conditions account for more hospital admissions than that of a new diagnosis of diabetes, cancer, coronary heart disease, heart failure, stroke, or HIV \citep{gale2014public}. Unlike general surgery procedures, which are scheduled and often done in an out-patient setting, EGS operations are typically in-patient procedures that occur after a chronic condition becomes acute. As such, treatment for EGS conditions is a critical component of emergency health care that differs from general surgery. 

In our study, we seek to assess the extent of racial disparities in adverse events following EGS treatment. To do so, we analyze a large administrative dataset based on all-payer hospital discharge claims from New York, Florida, and Pennsylvania in 2012-2013. We compare Black patients to white patients who are admitted for an EGS condition across 432 hospitals. In our analysis, we seek to to carefully adjust for baseline differences in the Black and white populations to understand the extent to which differences are: (1) attributable to patient-level risk factors versus hospital-specific factors, and (2) attributable to the decision to operate. 

Using statistical methods to assess disparities after adjusting for baseline differences between groups has a long history in health research and the social sciences \citep{fortin2011decomposition}. 
Linear regression is by far the most common adjustment approach; a specific implementation known as the \emph{Kitagawa-Oaxaca-Blinder} (KOB) decomposition is the workhorse method \citep{kitagawa1955components, Blinder1973, Oaxaca1973}; see \citet{sen2014using} and \citet{basu2015using} for examples of this approach in health research. Linear regression can perform poorly in practice, however, often relying on substantial extrapolation to adjust outcomes between groups \citep{sloczynski2020average}. Alternative methods like matching can avoid extrapolation \citep{silber2013characteristics}, but at the cost of poor balance in large administrative data sets.
The choice of what to adjust for is also critical. Standard approaches like the KOB decomposition often adjust for all available baseline covariates. Changing the adjustment set, however, can also lead to large changes in estimated disparities \citep{jackson2018decomposition}. 

Our methodological contribution is to develop statistical tools for estimating racial disparities using linear weighting estimators, which include both regression and matching as special cases. Building on a growing literature in causal inference and econometrics \citep[see][for a review]{benmichael2021_review}, we first characterize the error of a linear weighting estimator for the average surgery outcome for one group, adjusted for a selected set of baseline covariates to follow a particular target distribution. For example, we consider the average outcome for white patients adjusted to have the same baseline characteristics as Black patients.
We then show that a regularized form of the Oaxaca-Blinder estimator controls this error, and extend the decomposition to include hospital differences. Next, we show how to estimate the weights via a constrained optimization problem, which includes both matching and (penalized) regression as special cases. Importantly, we can constrain the weights to be non-negative, thus avoiding extrapolation. We can also directly control the bias-variance trade-off with a single tuning parameter in the optimization problem. We then show how to use our framework for estimands that decompose overall disparities based on the treatment decision, following recent proposals from \citet{jackson2018decomposition} and \citet{Yu2023_decomp}. For instance, we consider a counterfactual scenario in which Black patients have the same odds of undergoing surgery as white patients, conditional on a subset of covariates. This allows us to further decompose observed disparities in order to assess the role of the decision to perform surgery on patients admitted for an EGS condition.

We find that before adjustment, Black patients appear to have better outcomes than white patients. However, once we adjust for baseline covariates, we see significant disparities in the rate of adverse events and the length of stay for both the full patient population and the subset that receives operative care. Further, we find that adjusting for admitting hospital largely eliminates these estimated disparities, though Black patients still have longer lengths of stay in the hospital than re-weighted white patients. This provides additional evidence that variation across hospitals is an important driver of health disparities and that interventions targeted at hospital quality may be critical for further reducing racial disparities in surgical outcomes. Finally, we find that differences in the decision to treat EGS patients have little bearing on overall disparities.

Complementing our substantive results, we argue that three key features make the methodology we develop particularly suited to clinical applications such as the one we consider here.  First, many clinical applications use claims databases with large sample sizes. We show that our method is fast computationally, especially relative to matching, making it feasible with the large samples. Second, accounting for differences within clusters, such as hospitals, is a critical element of the research design for studying racial disparities in a range of applications. Our proposed approach naturally generates weights to balance comparisons within such clusters. Finally, with large numbers of patient covariates, there may be some subset of the covariates that interact. However, selecting relevant interactions is a high-dimensional variable selection problem that is separate to the estimation of the weights. We outline how to use sample splitting and a random forest to discover key interactions for adjustment. 

Our article proceeds as follows. In Section~\ref{sec:application_intro}, we outline the details of our dataset and study. In Section~\ref{sec:setup} we set up the statistical framework and review existing approaches. In Section~\ref{sec:bal_weights} we develop regularized linear estimators for adjusting covariate differences between a focal group and comparison group, and apply this to estimate racial disparities in our setting. In this section, we also discuss accounting for hospital-level interactions, expanding basis functions with interactions, and hyper-parameter selection.  In Section~\ref{sec:results}, we analyze the data on racial disparities in EGS care. In Section~\ref{sec:conc}, we conclude.

\subsection{Racial Disparities in Emergency General Surgery in PA, NY, FL}
\label{sec:application_intro}

Our study uses a dataset based on all-payer hospital discharge claims from New York, Florida and Pennsylvania in 2012-2013. We restricted the study population to all patients admitted for inpatient care emergently, urgently, or through the emergency department with a diagnosis of an acute general surgical condition. We classified acute general surgical condition types using a modified list of 124 International Classification of Diseases, Ninth Revision, Clinical Modification (ICD-9-CM) codes that represent the scope of emergency general surgery \citep{shafi2013emergency}. Importantly, patients with EGS conditions may or may not be treated with an operation. Our data set includes all patients with an EGS condition, and some subset of these patients had surgery.

All patients were classified into one of nine possible surgical conditions (resuscitation, general abdominal, upper gastrointestinal, colorectal, hernia, intestinal obstruction, hepatobiliary, skin and soft tissue, and vascular) and then further classified into 51 specific acute EGS conditions; the indicators for these 51 conditions are important variables for adjustment. Patient demographic and clinical characteristics were also abstracted from the claims datasets. First, we used Elixhauser indices to define 31 comorbidities \citep{elixhauser1998comorbidity}. Next, we developed a measure of patient frailty using a set of specific ICD-9-CM codes which represent clinical manifestations of frail patients in administrative data \citep{kim2014measuring}. We also defined an indicator for severe sepsis using the Angus implementation for severe sepsis algorithm \citep{angus2001epidemiology}. Table~\ref{tab:covs} in the appendix contains a full list of the baseline covariates. Our primary outcome is the presence of an adverse event within 30 days after admission. We define an adverse event as the presence of either death, a complication, or a prolonged length-of-stay. As a secondary outcome, we use hospital length-of-stay; we top-code this at 30 days, which affects 2.8\% percent of the sample. 

In our data, race and Hispanic ethnicity is determined from self-reported records. We restrict the sample to patients who self-report as either white or Black/African American and exclude patients that self-report as Hispanic. After these exclusion criteria, in our data, there are 133,174 Black patients and 556,331 white patients across 432 hospitals. In the data, almost 60\% of the patients undergo surgery as treatment for an EGS condition. Unadjusted, the risk of an adverse event following surgery for an EGS condition is 7\% lower for Black patients than for white patients (Risk ratio: 0.93, 95\% CI: 0.92, 0.94). In the unadjusted analysis, Black patients show no meaningful difference in terms of length of stay when compared to white patients.

These unadjusted differences in outcomes, however, might be driven in part by baseline differences in the health status of Black and white patients. There are critical differences across the two patient populations in our data. Black patients tend to be younger, are more likely to be enrolled in Medicaid, and have notably higher rates of diabetes and hypertension (see Table~\ref{tab:bal} in the appendix for these balance statistics). Moreover, the unadjusted differences in outcomes compare patients who received care at different hospitals. Thus, these unadjusted differences likely also capture differences in quality of care across hospitals, as well as possible unobserved differences in the surrounding areas, such as the availability of other services. For instance, \citet{silber2015examining} found no racial disparities in outcomes when using within-hospital comparisons for general surgery procedures.  The goal of our analysis is therefore to estimate differences in outcomes between Black and white patients after adjusting for observable baseline characteristics as well as the hospital in which the surgery was performed. In addition, there may be racial disparities in which patients are selected for treatment via surgery. As such, we consider three main analyses. First, we find weights that control the imbalance in patient-level baseline characteristics, co-morbidities, and surgery type between re-weighted white patients and Black patients within each state. Second, we find weights that also control this imbalance within each hospital. Third, we also perform a decomposition analysis that assesses whether disparities in outcomes are exacerbated by racial disparities in the decision to operate on EGS patients.

\section{Framework, Notation, and Assumptions}
\label{sec:setup}

There is a long history across health and the social sciences of quantifying group level disparities, and in trying to understand the extent to which observable characteristics explain these differences \citep{Yu2023_decomp}. For overview discussions, see, among others: \citet{fortin2011decomposition} in economics, \citet{jackson2020meaningful} in health, and \citet{lundberg2021gap} in sociology. In this section, we review this approach as applied to racial disparities in EGS outcomes. Next, we outline notation and review a series of different estimands that can describe disparities.

\subsection{A General estimand}
\label{sec:prelim}

We begin by formalizing a general estimand and reviewing the key assumptions required to  measure racial disparities that adjust for background characteristics. In our study, the patient population is indexed by $i=1,\dots,n$, and we denote patient race using a binary variable $G_i \in \{0,1\}$ where $G_i = 1$ indicates that the patient is Black and $G_i = 0$ that the patient is white. In addition to race, we observe a set of baseline covariates $\mathbf{X}_{i} \in \mathcal{X} \subset \R^d$ and an outcome $Y_i$. 
There are two particularly important variables in our analysis: (i) whether patient $i$ undergoes surgery, which we denote as $W_i \in \{0,1\}$,\footnote{While we have a well-defined set of EGS conditions and a set of patients that have all been diagnosed with one of these conditions, 40\% of patients that are diagnosed with these conditions do not receive surgery.} and (ii) which hospital treats patient $i$, which we denote as $H_i \in \{1,\ldots,J\}$.
We assume that the tuples of patient information $(X_i, G_i, Y_i)$ are sampled i.i.d. from some distribution $\mathcal{P}$, and we drop the patient subscript $i$ when convenient for notation. We let $n_g$ denote the number of patients in our sample with race $g$, and denote the sample average outcomes and covariates for patients of race $g$ as $\bar{Y}_g$ and $\bar{X}_g$, respectively. Similarly, we denote the $n_g \times d$ matrix of sample covariate values for group $g$ as $X_g \in \R^{n_g \times d}$.

Informally, our analysis is focused on measuring disparities in outcomes between patients across groups, while ``adjusting'' for different sets of covariates.
To formalize this, let $V_i \in \mathcal{V} \subseteq \mathcal{X}$ be a subset of the covariates.
The primary mathematical object in our analysis is the expected outcome, conditional on race $G$ and selected covariates $V$,  $m(v,g) \equiv \E[Y \mid V = v, G = g]$.
Our estimand is then the expected value of this conditional expectation when the chosen covariates $V$ are drawn from some distribution $\mathcal{P}_v^\ast$ and race $G$ is set to $g$:
\begin{equation}
  \label{eq:estimand}
  \mu^{\mathcal{P}_v^\ast}_g = \E_{V \sim \mathcal{P}_v^\ast}[m(V, g)] = \int m(v,g) d\mathcal{P}_v^\ast(v).
\end{equation}
This estimand follows the nonparametric framing proposed by \citet{Barsky2002} to decompose racial disparities in wealth, generalizing it to account for arbitrary target distributions $\mathcal{P}^\ast$; see also \citet{jackson2022observational}. This quantity measures what we would expect for patients of race $g$, restricted to the population where other patient-level information is distributed according to $\mathcal{P}^\ast$. 
Importantly, this target quantity is a \emph{statistical} estimand rather than a causal estimand, since it does not necessarily reflect a possible intervention.

There are two key components in defining the estimand $\mu_g^{\mathcal{P}^\ast_v}$. First and foremost, we must select the covariates $V$ in the adjustment set. This is the most important choice in the analysis. As we will show in our data below, disparities that appear when adjusting for one set of variables may not appear when adjusting for a different set of variables. Second, we must select the target distribution of the covariates to average over, $\mathcal{P}^\ast_v$, as this defines which distributions we are comparing. Most clinical studies of racial disparities focus on estimands of this form, and our analysis in Section~\ref{sec:results} uses several special cases of this general estimand. We describe these special cases and their potential limitations next.

\subsubsection{Standardizing white patients' covariates to the distribution of Black patients' covariates}
\label{sec:standard}

First, we consider an estimand commonly used in clinical studies \citep[e.g.,][]{silber2013characteristics,silber2014racial,silber2015examining,rosenbaum2013using}, that seeks to estimate the expected outcome for white patients as-if other selected patient-level covariates $V$ followed the distribution for Black patients.
We let $\mu^{g'_v}_{g}$ be shorthand notation for setting $\mathcal{P}^\ast_v$ to the conditional distribution of the selected covariates $V$ given $G = g'$, and then denote the target quantity as $\mu^{1_v}_0$. We compare this to $\mu^{1}_1$, the mean observed outcome for Black patients, which we estimate via the sample average. For continuous outcomes such as the length of stay, we make this comparison through the adjusted difference between Black and white patients $\mu^{1}_1 - \mu^{1_v}_0$. With binary outcome measures, we additionally consider the adjusted risk ratio $\mu^{1}_1 / \mu^{1_v}_0$.

The key assumption underlying this estimand and our ability to estimate $\mu^{1_v}_0$ non-parametrically is a notion of overlap:
\begin{assumption}[Overlap]
$ P(G = 1 \mid V = v) < 1$ for all $v \in \mathcal{V}$.
\end{assumption}
This overlap assumption rules out the possibility that there is some combination of patient-level information that can deterministically predict whether the patient is Black. This assumption ensures that for any Black patient in the population we can find a white patient with similar selected characteristics $V$, given a large enough sample.

There are many potential choices for the covariate set to adjust for, and we consider several in our analysis. First, we use \emph{all} of the patient-level covariates in our adjustment set using the full EGS patient population. Next, we use all the patient covariates in the subset that was selected for operative care ($W=1$). Comparing the magnitude of the estimated disparities in the full study population and the subset that receives surgery gives a sense of the relationship between the surgery decision and disparities. Second, we add hospital-level information in our adjustment set. One way to do so would be to include hospital-level covariates. This approach, however, does not account for any unmeasured (or unmeasurable) factors, such as features of the surrounding community or other local supports. Instead, we directly adjust for the hospital that treats patient $i$, $H_i$, and only compare patients within the same hospital. This accounts for all observable and unobservable differences across hospitals by restricting comparisons between racial groups to be within the same hospital and then marginalizing over the hospitals. 

Finally, we consider a more nuanced analysis that takes into account equity concerns, reflecting recent proposals from \citet{jackson2018decomposition}, \citet{jackson2020meaningful}, and \citet{jackson2022observational}.\footnote{This is closely related to the literature on ``good'' and ``bad'' controls; see, for example, \citet{cinelli2020controls}.}  In their framework, \emph{allowable} covariates are those for which adjusting for differences between groups is justified by clinical and ethical considerations; examples include age and sex. By contrast, factors that reflect socioeconomic status, such as insurance type, are \emph{non-allowable} for adjustment, in the sense that these factors contribute to health disparities. For example, differences in surgery decisions based on income would be considered unethical, and adjusting for these differences between groups would therefore mask important inequities.  

Specifically, \citet{jackson2020meaningful} proposes partitioning covariates into four different sets $X =  (X^n, W, X^y, X^w)$. The first two covariate sets are $X^n$, the non-allowable covariates, and $W$, the surgery decision. The allowable covariates are split into $X^y$ and $X^w$, the sets of ``outcome-allowable'' and ``treatment-allowable'' covariates, respectively. Outcome-allowable covariates are typically demographic measures --- in our analysis age and sex --- while treatment-allowable covariates are typically clinical measures related to health history, such as the presence of comorbidities, that inform the surgery decision. Slightly abusing notation, we denote $\mathcal{P}_{X^y}^1$ as the conditional distribution of the outcome-allowable covariates $X^y$ given that race $G = 1$. Then the estimand $\mu_0^{{1}_{X^y}}$, corresponds to the average observed outcome for white patients, standardizing the distribution of \emph{outcome-allowable covariates} to be equal to that for Black patients.

\subsubsection{Decomposing disparities}
\label{sec:decompose}

Thus far, we have considered estimating an health disparity by standardizing white patients' covariates to the distribution of Black patients. For this approach, we can either adjust for the full covariate vector $V$ or partition the covariates into different sets to account for equity concerns. One drawback to this approach is that it does not formally account for how the treatment decision may contribute to disparities. In our application, for example, clinicians may be interested in understanding how the decision to treat EGS patients via surgery affects disparities. Next, we consider two different methods that integrate the decision to treat into the disparity estimation process. Both of these methods decompose the estimated disparity into components that are based on the treatment decision.
These decompositions involve estimands that are based on frameworks with assumptions that allow for causal interpretations. Here, we consider them to be important and relevant \emph{statistical} estimands.

\paragraph{Stochastic intervention decomposition.}

First, \citet{jackson2018decomposition,jackson2020meaningful,jackson2022observational}  propose a decomposition based on a counterfactual scenario where Black patients have the same odds of undergoing treatment as white patients, conditional on outcome- and treatment-allowable covariates.\footnote{\citet{jackson2020meaningful} embeds these concepts into a potential outcomes framework, where there are two potential outcomes for each patient depending on whether or not they undergo surgery, and formalizes this estimand as the expected outcome under a stochastic intervention. Here we focus on the resulting \emph{statistical} estimand that is equivalent to the causal estimand under several identifying assumptions \citep{jackson2020meaningful}.} Using our notation, the resulting target distribution can be written as:
\[
  \mathcal{P}_{w | 0}(X = x) \equiv P(W = w \mid G = 0, X^y = x^y, X^w = x^w) P(X^y = x^y, X^w = x^w, X^n = x^n \mid G = 1),
\]
for $w \in \{0,1\}$, where $P(W = w \mid G = 0, X^y = x^y, X^w = x^w)$ is the probability that a white patient with allowable covariates $x^y$ and $x^w$ receives treatment. This distribution breaks the observed relationship among Black patients between the surgery decision $W$ and the  $X^y, X^w, X^n$ covariates. Instead, this estimand sets the distribution of the surgery decision to follow the observed relationship among white patients, with the other covariates based on the Black patient population.

With these quantities, we follow \citet{jackson2020meaningful} and define the following decomposition based on three different estimands:
\begin{equation}
  \label{eq:decomp_jackson}
  \underbrace{\mu_1^1 - \mu_0^{{1}_{X^y}}}_{X^y\text{-adjusted disparity}} = \underbrace{\mu_1^1 - \mu_1^{\mathcal{P}_{w|0}}}_{\text{disparity reduction}} + \underbrace{\mu_1^{\mathcal{P}_{w|0}} - \mu_0^{{1}_{X^{y}}}}_{\text{residual disparity}}.  
\end{equation}

On the left-hand side is the group disparity after adjusting for outcome-allowable covariates only, $\mu_1^1 - \mu_0^{{1}_{X^y}}$. In the context of our analysis, this is the disparity in adverse outcomes for the entire EGS population after adjusting for age and sex only. This disparity is then decomposed into two terms. 

The second term, $\mu_1^1 - \mu_1^{\mathcal{P}_{w|0}}$, compares the observed average outcome for Black patients $\mu_1^1$ with the average outcome if the odds of performing surgery are equalized across Black and white patients. This term is the \emph{disparity reduction} for Black patients, which measures how adverse outcomes for Black patients would change after equalizing the odds of treatment via surgery. The third term in the decomposition, $\mu_1^{\mathcal{P}_{w|0}} - \mu_0^{\mathcal{P}_{1}^{y}}$, compares outcomes for Black patients after equalizing the odds of surgery to  outcomes for white patients after adjusting for outcome-allowable covariates. This term is referred to as the \emph{residual disparity}, the amount of observed disparity that would remain after equalizing the odds of the surgery. If the outcome-allowable-covariate-adjusted disparity and residual disparity are close to each other this suggests that the decision to perform surgery is not a meaningful driver of observed disparities.

\paragraph{Unconditional decomposition.}
\citet{Yu2023_decomp} propose an alternative approach that decomposes the unadjusted disparity into four components based on: (i) the baseline disparity among Black and white patients if no one receives treatment; (ii) the difference in treatment prevalence; (iii) the effect of treatment between the two groups; and (iv) differential selection into treatment based on its effects.  Following their identifying assumptions, we can write the resulting statistical estimand in terms of the estimand in Equation~\eqref{eq:estimand}. Using the full vector of covariates, $V$, the target distribution is:
\[
\mathcal{P}_{wg}(X = x) \equiv \bbone\{W = w\} P(X = x \mid G = g).
\]
Using $\mu_g^{wg}$ as a shorthand for $\mu_g^{\mathcal{P}_{wg}}$, and denoting the rate that group $g$ receives surgery as $p_g \equiv P(W = 1 \mid G = g)$, we can then decompose the unadjusted disparity into the four components:
\begin{equation}
  \label{eq:decomp_yu}
  \underbrace{\mu_1^1 - \mu_0^0}_{\text{total}} = \underbrace{\mu_1^{01} - \mu_0^{00}}_{\text{baseline}} + \underbrace{(\mu_0^{10} - \mu_0^{00}) \times (p_1 - p_0)}_{\text{prevalence}} + \underbrace{p_1 \times (\mu_1^{11} - \mu_1^{01} - \mu_0^{10} + \mu_0^{00})}_{\text{effect}} + 
  \underbrace{\text{Cov}_0(W,\tau) - \text{Cov}_1(W, \tau)}_{\text{selection}}.  
\end{equation}
where $\tau$ represents the patient-level effect of surgery on the adverse event outcome, and $\text{Cov}_g(W,\tau)$ is the covariance between the decision to operate and the treatment effect of surgery for group $g$. These covariance terms will be positive if those patients that are more likely to benefit from surgery are more likely to receive surgery.  

In this decomposition, the ``baseline'' term refers to the average difference in adverse events if no patients in either group are treated via surgery. The ``prevalence'' term captures how much of the disparity in adverse events is due to any differential prevalence of surgery by race. The ``effect'' term reflects the possibility that the average effect of surgery differs by race. Finally, the ``selection'' term captures whether differential selection into the decision to operate contributes to the total disparity. This decomposition fully captures all aspects of how the disparity in outcomes might arise. In order to estimate each of these terms via weighting, we rearrange terms as follows:

\begin{equation*}
\label{eq:decomp_est}
    \text{Selection} = \text{Total} - \text{Baseline} - \text{Prevalence} - \text{Effect}
\end{equation*}

Finally, note that this four part decomposition is compatible with the decomposition proposed by \citet{jackson2020meaningful}: the decomposition in Equation~\eqref{eq:decomp_yu} can be viewed as a further decomposition of that in Equation~\eqref{eq:decomp_jackson}. See \citet{Yu2023_decomp} for further discussion.

\subsection{Review: adjustment via linear estimators}
\label{sec:ob}

We now use this setup to describe existing approaches to measure and isolate racial disparities, placing them in a common estimation framework. In this section, we will focus on estimating $\mu_0^{1_v}$, the expected outcome for white patients if the selected patient-level characteristics $V$ were drawn from the distribution of Black patients;
in Section~\ref{sec:bal_weights} we extend to a generic target population $\mathcal{P}^\ast$.
To begin, note that the average outcome among white patients will in general be a biased estimator for this quantity:
\[
  \E\left[\bar{Y}_0\right] - \mu_0^{1_v} = \E[m(V, 0) \mid G = 0] - \E[m(V, 0) \mid G = 1].
\]
This bias is due to the selected patient-level information $V$ varying systematically between white and Black patients, and is precisely the difference we are attempting to account for.

One classical approach to account for this difference is based on regression. The Kitagawa-Oaxaca-Blinder decomposition \citep{kitagawa1955components,Blinder1973, Oaxaca1973} separately regresses the outcomes $Y$ on the selected covariates $V$ for Black and white patients via ordinary least squares, yielding coefficients $\hat{\beta}_1$ and $\hat{\beta}_0$. It then estimates $\mu_0^{1_v}$ by taking the average prediction for Black patients, if we were to set the race variable to be white in the regression:
\begin{equation}
  \label{eq:ob_est}
  \hat{\mu}_0^{1_v KOB} \equiv \hat{\beta}_0 \cdot \bar{V}_1,
\end{equation}
where $\bar{V}_1$ is the average of the selected covariates for Black patients.

This regression approach is an imputation estimator: it uses the predictions from the regression to impute expected patient outcomes for Black patients if the race indicator $G$ changed but the selected covariates $V$ remained fixed. Variants of this approach are common in health disparities research \citep{alavi2012racial,esnaola2008race,cooper1996surgery}.
As \citet{Kline2011} discusses, this regression approach can also be viewed as a linear combination of outcomes for the white patient population,
\[
  \hat{\mu}_0^{1 KOB} = \sum_{G_i = 0} \hat{\gamma}_i^{KOB} Y_i, \text{  where  } \hat{\gamma}_i^{KOB} = V_i'(V_0'V_0)^{-1}\bar{V}_1,
\]
where $V_0$ is the matrix of selected covariates for white patients.
Typically, KOB-style analyses are also concerned with decomposing observed disparities; here
we focus on estimating the adjusted difference component of this decomposition, while retaining the terminology as in \citet{Kline2011}. 

Other common estimators of adjusted racial differences are similarly linear in outcomes, using specific types of regression modeling, inverse propensity weighting, or matching to estimate $\mu_0^{1_v}$. For some of the many examples of related literature, see: \citet{Barsky2002,dinardo1995labor,silber2014racial,silber2013characteristics,rosenbaum2013using,silber2015examining,chernozhukov2013inference, firpo2018decomposing, machado2005counterfactual,melly2005decomposition,jackson2020meaningful}.

\section{Regularized linear estimators for adjustment}
\label{sec:bal_weights}

In this section we extend the KOB approach for a general comparison distribution $\mathcal{P}_v^\ast$.
Specifically, we define the KOB estimate as $\hat{\mu}_g^{\mathcal{P}_v^\ast KOB} \equiv \hat{\alpha}_g + \hat{\beta}_g \cdot \E_{\mathcal{P}^\ast_v}[V]$, where $\hat{\alpha}_g$ and $\hat{\beta}_g$ are the solutions to the ordinary least squares problem
\[
  \min_{\alpha, \beta} \ \frac{1}{n_g}\sum_{G_i = g} \left(Y_i - \alpha - \beta \cdot V_i\right)^2.
\]
This builds on a long literature generalizing linear regression estimators in similar contexts; see \citet{strittmatter2021gender} for a related discussion applied to the gender pay gap in Switzerland.
While we focus here on adjusting for the set $V$, this estimator accommodates many types of analyses, as we describe in Section \ref{sec:ob}


We first consider the statistical performance of this KOB estimator in non-parametric settings where the linear model might not be correct. We then propose a ridge-regularized version that accounts for interactions and non-linearities in the outcome model. 
The most important interaction is that between hospitals and the covariates; we incorporate this by fitting separate models for each hospital and partially pooling.\footnote{The interaction between the surgery decision and the covariates is also important. The estimates we report in our analysis in Section~\ref{sec:results} use the surgery decision in different ways, where appropriate we run separate analyses for the population that do and do not receive surgery to account for this interaction.} Then we  connect this ridge regularized approach to a general linear weighting estimator. 

\subsection{A regularized Kitagawa-Oaxaca-Blinder approach}
\label{sec:regularized_ob}

\subsubsection{Statistical performance of KOB} 
One limitation of the Kitagawa-Oaxaca-Blinder approach outlined in Section \ref{sec:ob} is that it only accounts for covariates linearly, and so does not take into account potential non-linearities and interactions between the covariates $V$. These interaction terms may be important predictors of patient outcomes, and failing to account for these may lead to irreducible error in our estimates. Formally, we can characterize the design-conditional mean square error of the KOB estimator in general nonparametric settings with comparison distribution $\mathcal{P}^\ast_v$.

\begin{proposition}[MSE for Kitagawa-Oaxaca-Blinder]
\label{prop:ob_mse}
  Assume that the covariates are centered so that $\frac{1}{n_g}\sum_{G_i = g} V_i = 0$, and that $\Var(Y \mid V = v, G = g) = \sigma^2$ for all $x,g$. Then the design-conditional mean square error is
  \[
    \E\left[\left(\hat{\mu}_g^{\mathcal{P}_v^\ast KOB} - \mu_g^{\mathcal{P}^\ast_v}\right)^2 \mid V, G\right] = \underbrace{d_\text{eff} \ \frac{\sigma^2}{n_g}}_\text{estimation error} + \underbrace{\left(\E_{\mathcal{P}_v^\ast}\left[V\right]\Sigma_g^{-1}\rho_g - \E_{\mathcal{P}_v^\ast}[m(V, g) - \bar{m}_g]\right)^2}_\text{approximation error},
  \]
  where $\Sigma_g = \frac{1}{n_g}\sum_{G_i = g} V_i V_i'$, $\rho_g = \frac{1}{n_g} \sum_{G_i = g} V_i (m(V_i, g) - \bar{m}_g)$, $\bar{m}_g = \frac{1}{n_g}\sum_{G_i = g} m(V_i, g)$, and $d_\text{eff} = 1 + \frac{1}{n_g}\sum_{G_i = g} \left(\E_{\mathcal{P}_v^\ast}[V] \Sigma_g^{-1} V_i\right)^2$.
\end{proposition}

Proposition~\ref{prop:ob_mse} shows that the MSE for the Kitagawa-Oaxaca-Blinder estimator consists of two terms.\footnote{In heteroskedastic settings where the residual variance changes with $v$, Proposition~\ref{prop:ob_mse} can be written as an upper bound with the maximum residual variance taking the role of $\sigma^2$.}
First, because the intercept $\hat{\alpha}_g$ and coefficients $\hat{\beta}_g$ are unbiased estimates of the best linear approximation to the true conditional expectation $m(v,g)$, the estimation error corresponds to the variance of $\hat{\mu}_g^{\mathcal{P}_v^\ast KOB}$ with a correctly specified outcome model. As we expect, this variance decreases with the number of patients of race $g$, but there is an additional design effect factor $d_\text{eff}$ that depends on the distribution of the covariates $V$ both in the sample and under $\mathcal{P}_v^\ast$. If $\mathcal{P}_v^\ast$ is far from the empirical distribution of patients of race $g$ --- as measured by the mean of $V$ --- the resulting KOB adjustment will be large, leading to a larger estimation error. Conversely, if the distributions are close, there will be less need for adjustment and the estimation error will be closer to ordinary regression error.

While the estimation error decreases with the sample size, the second component of the MSE --- the approximation error --- does not. This term depends on the quality of the linear approximation to the true conditional expectation. If there are important non-linearities or interactions between covariates then this approximation error can be large.  We turn next to accounting for these terms in order to reduce the approximation error and improve the credibility of the KOB estimate.

\subsubsection{Accounting for hospital components and non-linearity} 

As we discuss in Section \ref{sec:prelim} above, a key element in studies of racial disparities in health is accounting for the admitting hospital.
To do so, we restrict comparisons between racial groups to be within the same hospital, which accounts for both observable and unobservable differences across hospitals. 

To formalize this estimator, we modify the KOB regression above in the case where the hospital $H$ is included as a selected covariate.
First, we decompose the selected covariates into $V_i = (H_i, Z_i)$, where $H_i \in \{1,\ldots,J\}$ is the hospital that treats patient $i$, and $Z_i$ encodes the other selected patient-level characteristics. We then consider an approximation to the conditional expectation with two key components: (1) separate hospital-level fixed effects and coefficients; and (2) further non-linearity and interactions in the other patient characteristics:
\begin{equation}
  \label{eq:interact_model}
  m(v,g) \approx \sum_{j=1}^J \bbone\{H_i = j\}\left(\alpha_{gj} + \beta_{gj} \cdot \phi(Z_i)\right),
\end{equation}
where $\phi(z) \in \R^p$ encodes transformations of the patient characteristics.
In Section~\ref{sec:basis} we discuss how we choose this transformation in a data-driven way for our analysis.

The model in Equation \eqref{eq:interact_model} is  more flexible than the typical linear model:
it allows for separate hospital-specific models, where the relationship between the
patient-level characteristics and the outcome can differ by hospital.
Furthermore, by using an appropriate transformation function $\phi(\cdot)$ we can reduce the approximation error in Proposition \ref{prop:ob_mse}. At the same time, this model can be much more difficult to estimate due to the hospital-level coefficients. For example, in our data there are 406 separate hospitals, leading to 45,066 coefficients in Equation \eqref{eq:interact_model} when we use the full set of patient characteristics, even when $\phi(\cdot)$ is the identity transformation. This can lead to a high estimation error when estimating the model with OLS.

Instead, we choose to take a regularized approach. First, we find coefficients via regularized regression, choosing $\hat{\alpha}_{gj}$ and $\hat{\beta}_{gj}$ to solve
\begin{equation}
  \label{eq:ridge_model}
  \min_{\alpha_g, \beta_g, \mu_{\beta_g}} \frac{1}{n}\sum_{G_i = g}\left(Y_i - \sum_{j=1}^J \bbone\{H_i = j\}\left(\alpha_{gj} + \beta_{gj} \cdot \phi(Z_i)\right)\right)^2 + \lambda \sum_{j = 1}^J \|\beta_{gj} - \mu_{\beta_g}\|_2^2.
\end{equation}

Then we compute a regularized KOB estimate,
\begin{equation}
  \label{eq:reg_ob_est}
  \hat{\mu}_{g}^{\mathcal{P}_v^\ast RKOB} \equiv \sum_{j=1}^JP_v^\ast(H = j)\left(\hat{\alpha}_{gj} + \hat{\beta}_{gj} \cdot \E_{\mathcal{P}_v^\ast}[\phi(Z) \mid H = j]\right),
\end{equation}
where $P_v^\ast(H = j)$ is the marginal probability that hospital $j$ treats a patient under distribution $\mathcal{P}_v^\ast$.

To estimate the coefficients via Equation \eqref{eq:ridge_model}, we regularize the hospital-specific coefficients $\beta_{gj}$ towards a global model $\mu_{\beta_g}$. This allows us to share information across hospitals, potentially improving estimation. The level of pooling $\lambda$ is an important hyperparameter that controls the bias-variance tradeoff. As $\lambda$ decreases, less pooling occurs; when $\lambda = 0$ no information is shared across hospitals. This will have the lowest bias, as there is no regularization at all, but can increase variance. 
Conversely, as $\lambda$ increases, the hospital-specific coefficients become more and more similar; the extreme case where $\lambda \to \infty$ constrains the covariate-outcome relationship to be the same across hospitals, and only allows for baseline levels to differ. This substantially reduces the variance, as it removes the need to estimate separate models for each hospital, but can lead to high bias by ignoring possibly important interactions.
We discuss choosing $\lambda$ based on an alternative characterization of the estimator in the next section.
We inspect the impact of the level of pooling in our simulation study in Section \ref{sec:sim} and in our analysis in the supplementary materials.

\subsection{An equivalent penalized linear estimator}
\label{sec:balancing}
Following our discussion in Section \ref{sec:ob}, we now relate the regularized KOB approach above to a general linear estimator for $\mu_g^{\mathcal{P}_v^\ast}$. These take a linear combination of outcomes in group $G = g$,
\begin{equation}
  \label{eq:general_estimator}
  \hat{\mu}_g^{\mathcal{P}_v^\ast} \equiv \sum_{G_i = g} \hat{\gamma}_i Y_i,
\end{equation}
where $\hat{\gamma}_i$ is the weight in the linear combination placed on patient $i$; these weights may be positive or negative.\footnote{We restrict our attention to linear combination weights that are independent of the outcomes.} Next, we inspect the statistical properties of linear estimators and then find the MSE-optimal linear combination via convex optimization. We then show that the regularized KOB estimator in the previous section is a special case.

\paragraph{Estimation error of linear estimators.}
We begin by inspecting the statistical properties of the linear estimator in Equation \eqref{eq:general_estimator}, following the development in, e.g. \citet{hirshberg2019minimax,Hirshberg2019_amle}.
We can decompose the design-conditional mean square error  of $\hat{\mu}_g^{\mathcal{P}_v^\ast}$ into a bias term and a variance term:
\begin{equation}
    \label{eq:mse}
     \E\left[\left(\hat{\mu}_g^{\mathcal{P}_v^\ast}- \mu_g^{\mathcal{P}_v^\ast}\right)^2 \mid V, G\right] = \underbrace{\left(\sum_{G_i = g} \hat{\gamma}_i m(V_i, g) - \E_{\mathcal{P}_v^\ast}\left[m(V, g)\right]\right)^2}_{\text{bias}^2} + \underbrace{\sum_{G_i = g} \hat{\gamma}_i^2 \E[(Y_i - m(V_i, G_i))^2\mid V_i, G_i]}_{\text{variance}}. 
\end{equation}
The bias term comes from imbalance in the conditional expectation $m$. The variance term is
 the sum of squared weights, weighted by the residual variance for each patient.

Focusing on the flexible model with hospital-level interactions in Equation \eqref{eq:interact_model}, we can show that the bias depends on the difference between the target average of the (transformed) covariates (i) \emph{within} each hospital and (ii) \emph{across} hospitals, as well as the difference in hospital counts:
\begin{equation}
  \label{eq:bias_interact_model}
  \begin{aligned}
    \text{bias} & =  \underbrace{\sum_{j=1}^J \alpha_{gj} \left(\sum_{G_i = g, H_i = j}\hat{\gamma}_i - P^\ast_v(H = j)\right)}_{\text{difference in hospital count}}\\
    & \qquad + \underbrace{\sum_{j = 1}^J \left(\beta_{gj} - \bar{\beta}_g\right) \cdot \left(\sum_{G_i = g, H_i = j}\hat{\gamma}_i \phi(Z_i) -  P_v^\ast(H = j)\E_{\mathcal{P}_v^\ast}[\phi(Z) \mid H = j]\right)}_{\text{difference within hospitals}}\\
    & \qquad + \underbrace{\bar{\beta}_g \cdot \left(\sum_{G_i = g} \hat{\gamma}_i \phi(Z_i) - \E_{\mathcal{P}_v^\ast}[\phi(Z)]\right)}_{\text{difference across hospitals}}.
  \end{aligned}
\end{equation}
This expression shows the importance of within-hospital comparisons.
Only by controlling the differences in patient-level characteristics \emph{within} hospitals can we adjust for hospital-level
differences in how these characteristics correlate with outcomes.
From this, we see that we can evaluate the potential for bias in a linear estimator by inspecting these three differences, while we can evaluate the variance via the sum of the squared weights. We will now use this as a guide to construct weights that control the bias and the variance.

\paragraph{Penalized optimization problem.}
Following proposals from e.g. \citet{zubizarreta2015stable,hirshberg2019minimax}, we find 
linear combination weights $\hat{\gamma}$ by solving the following convex optimization problem:
\begin{equation}
  \label{eq:opt_prob}
  \begin{aligned}
    \underset{\gamma}{\text{min}} \ & \sum_{j=1}^J \left\|\sum_{G_i = g, H_i = j}\gamma_i \phi(Z_i) - P_v^\ast(H = j)\E_{\mathcal{P}_v^\ast}[\phi(Z) \mid H = j]\right\|_2^2 + \lambda \sum_{G_i = g} \gamma_i^2\\
    \text{subject to}~~& \sum_{G_i = g} \gamma_i \phi(Z_i) = \E_{\mathcal{P}_v^\ast}[\phi(Z)]\\
    & \sum_{G_i = g, H_i = j}\gamma_i = P_v^\ast(H = j)\\
    & L \leq \gamma \leq U.
  \end{aligned}
\end{equation}
This optimization problem adapts that of \citet{benmichael2021_lor}, developed for estimating subgroup effects in observational studies, to our racial disparities setting. It has several components relating to the bias and variance terms described above. First, the objective includes the mean square within-hospital difference in the transformed covariates $\phi(Z)$; this corresponds to the within-hospital component of the bias. The other term in the objective is the sum of the squared weights, a proxy for the variance of the linear estimator. The constraints in the optimization problem ensure that the overall difference between the covariates across hospitals is zero, and that the overall amount of weight placed on each hospital is equal to the probability that a patient is treated in the hospital in the target distribution. These terms control the other two components of the bias.

In this optimization problem, the weights are regularized in two ways: (i) the variance penalty in the objective with hyper-parameter $\lambda$ and (ii) the constraint that weights are bounded between $L$ and $U$. To understand the relative contribution of these two terms, first consider the case where $L = -\infty$ and $U = \infty$, which we refer to as ``unrestricted'' balancing weights. In Appendix \ref{sec:proofs} we show that the solution to optimization problem \eqref{eq:opt_prob} in this case is equivalent to the linear weights implied by the regularized KOB estimator~\eqref{eq:reg_ob_est}, with $\lambda$ serving the same role of negotiating the bias-variance tradeoff in both optimization problems. 

Alternatively, we can constrain the weights to be non-negative (i.e. $L = 0, U = \infty$), which we refer to as ``restricted'' balancing weights. With this constraint, the weighted average of the covariates for white patients is restricted to be in the convex hull of all the observed covariates for white patients. Most prominently, this ensures 
that the final estimate is ``sample bounded'' between the minimum and maximum possible outcome value.
Viewed in this light, we see that unrestricted balancing weights (i.e., outcome modeling) can \emph{extrapolate} from the data. This allows for potentially decreased MSE at the cost of increased model dependence \citep{king2006dangers}. 
We consider the impact of these constraints in Section \ref{sec:sim} via simulations.  See \citet{benmichael2021_ascm, Chattopadhyay2021} for more discussion on weighting representations of regression methods. 

The MSE decomposition in Equation~\eqref{eq:mse} also gives a guide for choosing the hyper-parameter $\lambda$. Since $\lambda$ penalizes the sum of the squared weights, it corresponds to the variance term in the MSE. In the special case of homoskedastic errors --- i.e., when $\E[(Y - m(V, G))^2 \mid V, G] = \sigma^2$ --- the variance term simplifies to $\sigma^2$ times the sum of the squared weights. Therefore, we choose $\lambda$ by first regressing the outcome $Y$ on the selected covariates $V$ for racial group $G = g$, then setting $\lambda$ to be the residual variance. This heuristic choice of $\lambda$ places our estimates on one spot of the bias-variance tradeoff, in Appendix \ref{sec:hyperparameter} we discuss tracing out this tradeoff in our data.

Finally, we can view this optimization problem as a continuous relaxation of matching patients of different races, with an exact match on hospitals and a fine balance constraint marginally on each covariate \citep{Yiu2018}. This further highlights the approach's computational efficiency, especially relative to matching.
In particular, optimization problem \eqref{eq:opt_prob} is a Quadratic Program (QP) that is strongly convex for $\lambda > 0$. Traditionally, QPs have been solved with interior point methods, but recent first-order methods using the Alternating Direction Method of Multipliers \citep[ADMM;][]{Boyd2010} have been particularly successful at efficiently solving large QPs. In particular, \citet{osqp} develop OSQP, an efficient ADMM procedure for QPs. The OSQP procedure is particularly efficient in our setting because the within-hospital difference measures induce sparsity that OSQP can leverage via efficient sparse numerical linear algebra libraries \citep[see][for extensive numerical results showing efficiency in big data problems]{kim2022_small_weights}. Compared with matching approaches that solve mixed integer programs, this convex formulation allows us to quickly find the linear combination weights even with the large administrative data sets often found in health-services research. In our analysis in Section~\ref{sec:results}, we compare the computational performance of matching versus this QP approach to balancing weights.

\subsection{Estimation}
\label{sec:estimation}

\paragraph{Point estimation.} We now focus our discussion on how to estimate the EGS racial disparities quantities in Section~\ref{sec:prelim}, which are based on three main estimands. First, all of the disparity measures involve the expected outcome for Black patients,  $\mu^{1}_1$.
We estimate this simply via the sample average outcome among Black patients: $\hat{\mu}_1^1 \equiv \frac{1}{n_1}\sum_{G_i = 1} Y_i$. Second, we are interested in estimating $\mu_{0}^{1_v}$, the expected outcome for white patients as-if the selected covariates $V$ followed the distribution of Black patients. In our analysis in Section~\ref{sec:results} we consider several different covariate sets: all patient-level covariates, all patient-level covariates and the hospital indicator, and outcome-allowable covariates only. We estimate this using the regularized KOB approach, using the empirical distribution of the Black patient population in place of the population distribution.

Finally,  we want to estimate the components of the disparity decompositions discussed in Section~\ref{sec:decompose}.
For the stochastic intervention decomposition of \citet{jackson2020meaningful}, we want to estimate $\mu_1^{\mathcal{P}_{w|0}}$, the expected outcome for Black patients if, given outcome- and treatment-allowable covariates, Black patients had the same odds of receiving treatment as white patients. To estimate this via the regularized KOB approach, we again use the empirical distribution of the covariates among Black patients in place of the population distribution $P(X^y = x^y, X^w = x^w, X^n = x^n \mid G = 1)$. We then estimate $P(W = 1 \mid G = 0, X^y, X^w)$, the probability of receiving surgery for white patients given the allowable covariates, by fitting a logistic regression on a set of flexible features using the observed data for white patients. Combined, this gives us an estimate of the target distribution. We note that, unlike for the other targets estimands, we rely on parametric modeling to estimate key components of intervention-style estimands like this.

For the unconditional decomposition of \citet{Yu2023_decomp}, we want to estimate $\mu_g^{wg}$, the expected outcome for group $g$ if everyone were to receive surgery decision $w$. Here we use the regularized KOB approach, re-weighting the observed patients of group $g$ with surgery decision $w$ towards the empirical distribution of the observed covariates for all patients of group $g$. We also estimate $p_g$, the proportion of patients in group $g$ receiving surgery, with the empirical average.

\paragraph{Variance estimation.}
Next we estimate the variance of these estimates. For $\hat{\mu}_1^1$, this is a simple average of observed outcomes, so we estimate the variance as
$
\hat{V}_1^1 = \frac{1}{n_1^2}\sum_{G_i = 1} \left(Y_i - \hat{\mu}_1^1\right)^2.
$
To estimate the variance of the other estimates, note that the variance term in Equation \eqref{eq:mse} involves the conditional variance of the outcome given the appropriate set of covariates. For each estimate we begin by fitting a model for the conditional expected outcome given the selected covariates $V$ and race $G$, $\hat{m}(v, g)$ and get the predicted value for each unit $\hat{m}_i = \hat{m}(V_i, G_i)$. We then estimate the variance as 
\[
\hat{V}\left(\hat{\mu}_g^{\mathcal{P}_v^\ast}\right) = \frac{1}{\left(\sum_{G_i = g} \hat{\gamma}_i\right)^{2}}\sum_{G_i = g} \hat{\gamma}_i^2 \left(Y_i - \hat{m}_i\right)^2 = \sum_{G_i = g} \hat{\gamma}_i^2 \left(Y_i - \hat{m}_i\right)^2,
\]
where we have used the constraint in Equation \eqref{eq:opt_prob} that $\sum_{G_i = g} \gamma_i = 1$. We term this the residualized variance estimator (RVE). In our main results below, we fit the model via ridge regression, fully interacting the hospital indicators $H_i$ and the surgery decision $W_i$, when appropriate, with the transformed covariates $\phi(Z_i)$ and weighting according to $\hat{\gamma}$. \citet{hirshberg2019minimax} and \citet{benmichael2021_review} provide technical conditions for $\hat{\mu}_g^{\mathcal{P}_v^\ast}$ to be asymptotically normally distributed around $\mu_g^{\mathcal{P}_v^\ast}$, and for $\hat{V}\left(\hat{\mu}_g^{\mathcal{P}_v^\ast}\right)$ to be consistent for the true variance. Notably, the bias term due to any remaining differences in the distributions between the adjusted $G = g$ population and the target distribution must be small. We explore the impact of this bias on coverage of confidence intervals in Section \ref{sec:sim} below.

Finally, note that when we fit the conditional expectation with a constant model, $\hat{V}\left(\hat{\mu}_g^{\mathcal{P}_v^\ast}\right)$ is equivalent to the HC2 standard error estimate for the coefficient on race $G = g$, in a weighted least squares regression of the outcome on an intercept and race. This form of variance estimate is common for propensity score estimators \citep{ramsey2019using,pirracchio2016propensity}, though \citet{reifeis2022variance} suggest that the resulting variance estimator can perform poorly in practice. We assess this via simulation in Appendix~\ref{sec:coverage_sim} and find that the HC2 standard error is conservative in our simulation settings.

With these variance estimates in hand, we can construct approximate $1-\alpha$ confidence intervals for our racial disparity measures.
For example, for  the adjusted risk difference when standardizing covariates to the Black patient distribution, we estimate the squared standard error as $\hat{V}_\text{diff} = \hat{V}_1^1 + \hat{V}\left(\hat{\mu}_0^{1_v}\right)$ and create an approximate confidence interval via $\hat{\mu}_1^1 - \hat{\mu}_0^{1_v} \pm z_{1-\alpha/2}\sqrt{V_\text{diff}}$, where $z_{1-\alpha/2}$ is the $1-\alpha/2$ quantile of a standard normal variable. We also create a confidence interval for the log adjusted risk ratio via the delta method, estimating an overall squared standard via $\hat{V}_\text{risk} = \left(\frac{1}{\hat{\mu}_1^1}\right)^2 \hat{V}_1^1 + \left(\frac{1}{\hat{\mu}_0^{1_v}}\right)^2\hat{V}\left(\hat{\mu}_0^{1_v}\right)$ and computing the confidence interval as $\log\hat{\mu}_1^1 - \log \hat{\mu}_0^{1_v} \pm z_{1-\alpha/2}\sqrt{V_\text{risk}}$.
Finally, the disparity decompositions are linear combinations of the various $\hat{\mu}_g^{P^\ast_v}$ estimates, so we can compute their standard errors analogously.\footnote{To compute the standard error for the components of the unconditional decomposition, we treat $p_g$, the proportion of patients in group $g$, as known given our large sample sizes. See \citet{Yu2023_decomp} for a form of the standard error that incorporates uncertainty in $p_g$.}

\section{Racial Disparities in Emergency General Surgery}
\label{sec:results}

We now turn to our analysis of racial disparities in EGS.\footnote{In the supplementary materials, we conduct two different simulation studies to evaluate our proposed methods and better explore their properties.}
First, we discuss various data processing steps, including expanding the basis to account for non-linearity and interactions. We estimate two different sets of weights: (1) adjusting for age and sex only, and (2) adjusting for the full set of available patient-level covariates. In addition, we estimate one set of weights that ignores hospital information and one that stratifies by hospitals. This allows us to account for the role of hospital quality, which has been shown in past research to be an important factor in racial disparities in surgery. We then review how well the weights reduce imbalances between the comparison groups and  explore how constraining the weights to be non-negative prevents extrapolation. After these diagnostic checks, we estimate outcome disparities across the different adjustment sets, finding that disparities are largely explained by differences between hospitals. Here, we provide a direct comparison to widely used methods based on matching \citep{silber2013characteristics,silber2014racial,silber2015examining,rosenbaum2013using}. Finally, we implement an analysis for both of the disparity decompositions outlined in Section~\ref{sec:decompose}. In this analysis, we implement both the stochastic intervention decomposition and the unconditional decomposition to understand whether the decision to operate by race also contributes to the observed disparity.

\subsection{Data Processing}

We start our analysis by focusing on the full set of available covariates and hospital indicators. In Section \ref{sec:allowable_decomp} below, we partition variables into allowable and non-allowable sets, following the setup in Section \ref{sec:prelim}.

\paragraph{Selecting the features to balance.}
\label{sec:basis}
Under the regularized KOB approach, we can account for possible nonlinearity or key interactions between covariates by expanding the set of features as encoded by $\phi(\cdot)$. To account for possible nonlinearities, a spline basis can be applied to each continuous covariate. In our study, age is the most important continuous covariate, and we include a spline basis for age in $\phi(\cdot)$. For binary covariates, the primary concern is identifying which relevant interactions between covariates to consider \citep[see][for more discussion on the role of interactions]{benmichael2021_drp}. In principle, expanding the basis to include interactions can be based on substantive expertise. In practice, however, substantive expertise may be silent when a large number of interactions are possible: in our data set, there are 4,186 potential two-way interactions alone. 

As such, we use a random forest-based approach for selecting interactions using a method outlined in \citet{inglis2022visualizing}. See \citet{wang2021flame} for a more general discussion of learning appropriate bases. Among machine learning methods, random forests are especially well suited for selecting interactions \citep{wright2016little}. As such, we use an ``honest'' implementation of random forests based on sample splitting. We first split the data into a training and analysis sample. For this step, we randomly sample 2.5\% of the patients as the training sample. Using the training sample, we train a random forest to predict the outcome, select the set of interactions with the largest variable importance metrics, and then include these interactions in our transformation function $\phi(\cdot)$. We then discard the training sample from further analysis. We used separate random forest fits for the different covariate sets; we first focus on the full set of covariates, and then consider outcome-allowable and treatment-allowable sets in Section \ref{sec:allowable_decomp}.

\paragraph{Hyperparameter Selection.}
Finally, we select the hyperparameter $\lambda$, which controls the bias-variance tradeoff using the method outlined in Section \ref{sec:balancing}. We choose separate hyperparameters by surgery status using the full set of covariates and for age and sex alone; the hyperparameter values range from 0.13 to 0.19. Notably these hyperparameter values are non-zero but still relatively small, so these values largely prioritize bias reduction over increasing the effective sample size. In the supplemental materials, we include an additional analysis where we compute the imbalance reduction and the effective sample sizes for a range of hyperparameter values. In that analysis, we demonstrate the clear bias-variance trade off controlled by the hyperparameter value.

\subsection{Diagnostics}

\paragraph{Balance.}
One important diagnostic for a weighting analysis is understanding how well the weights account for imbalances across the comparison groups. Here, we report on balance diagnostics to understand how well restricted weights (i.e. $L = 0, U = \infty$) balance the baseline differences between the white and Black patient populations. In the supplementary materials, we report balance results for the weights that do not include hospitals. For these analyses, we find that we can nearly exactly balance patient-level covariates. Here, we focus on the analysis that includes the full set of covariates $X$ and the hospital indicators $H$. 

To measure balance within hospitals, we calculate two vectors of standardized differences within each hospital before and after weighting, which we denote, respectively, as $\hat{\Delta}_{huw}$ and $\hat{\Delta}_{hw}$. For $K$ covariates with $H$ hospitals, the dimension of these vectors of standardized differences is now $K \times H = J$. We then calculate a measure we call the root-mean squared imbalance (RMSI) $= \left[\sum_h \hat{\Delta}_{hw}^2 \right]^{1/2}$. Importantly, RMSI avoids ``averaging away'' imbalances across hospitals. In Figure~\ref{fig:bal.plot}, we plot the RMSI for the 15 covariates with the largest imbalances at baseline. We observe, as anticipated, that the weighting estimator reduces these imbalances, often substantially, though some interactions remain difficult to balance.

\begin{figure}[htbp]
  \centering
    \includegraphics[scale=0.6]{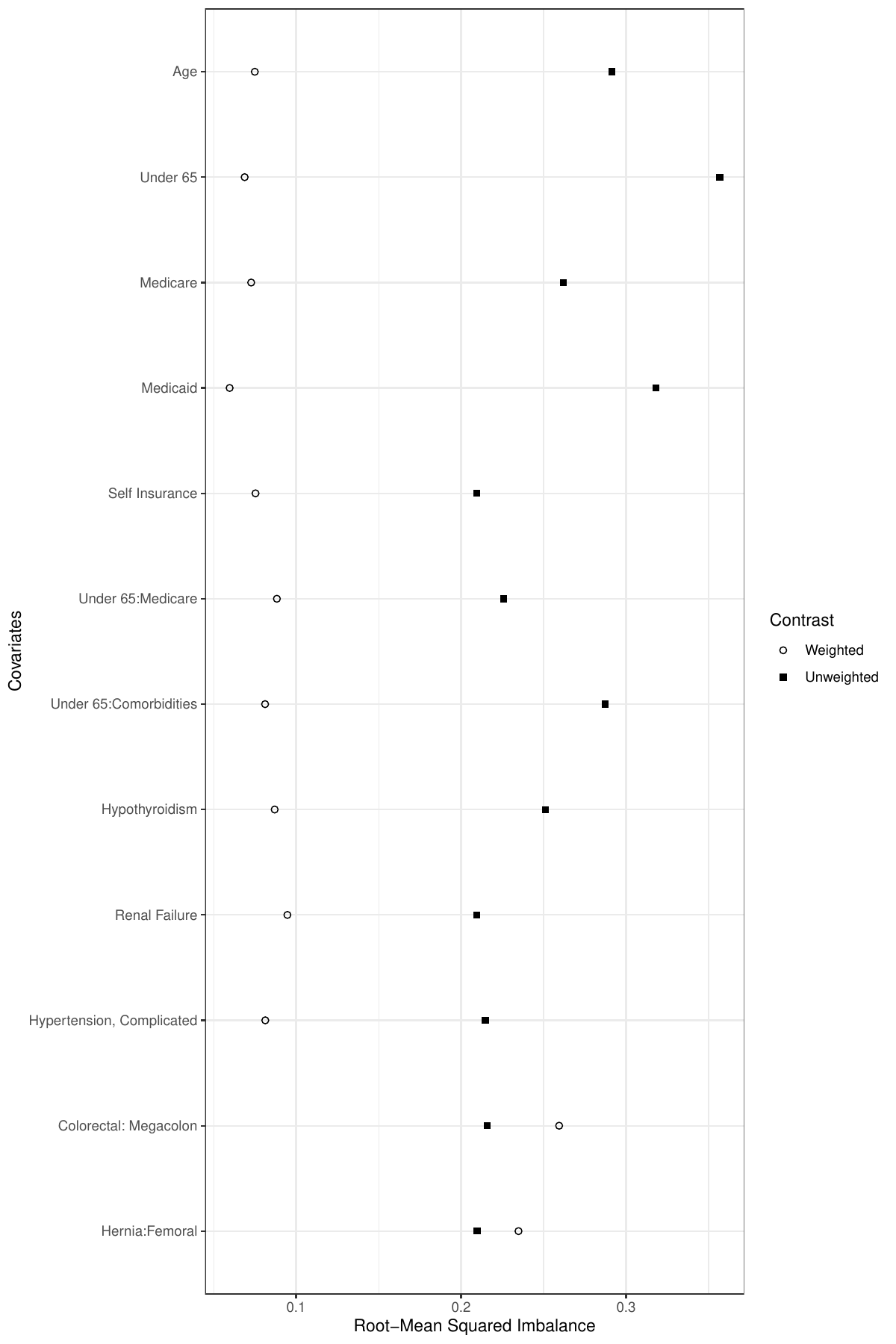}
     \caption{RMSI before and after weighting for the covariates with the largest baseline imbalances for the within hospital analysis.}
  \label{fig:bal.plot}
\end{figure}

\paragraph{Extrapolation.} 

We also investigate how constraining the weights to be non-negative prevents extrapolation. To that end, we estimated an additional set of weights using an unrestricted model where $L = -\infty$ and $U = \infty$ for the within-hospital analysis using the full set of covariates; this is equivalent to estimating a ridge regression outcome model alone. The unrestricted weights are allowed to be negative, which extrapolates beyond the support of the data. To measure the amount of ``effective extrapolation,'' we sum over absolute values of the negative weights and divide that sum by the number of Black patients. By this measure, 15.2\% of weights are negative overall. We also calculated this measure of effective extrapolation for each hospital. Figure~\ref{fig:hist2} contains a histogram of the effective extrapolation measure separately by hospital. The mass of the distribution is to the left of 50\%, which indicates that for most hospitals, the effective extrapolation is less than half. However, there is a substantial fraction of hospitals with more substantial extrapolation. In fact, there are several hospitals where the extrapolation percentage is greater than 100\%. Overall, we find that, for many hospitals, considerable extrapolation is necessary for balancing covariates with unrestricted weights --- and an extreme amount of extrapolation is necessary for a few. 

\begin{figure}[tb]
\centering
\begin{subfigure}[t]{0.48\textwidth}
  \centering
    \includegraphics[width=0.95\textwidth]{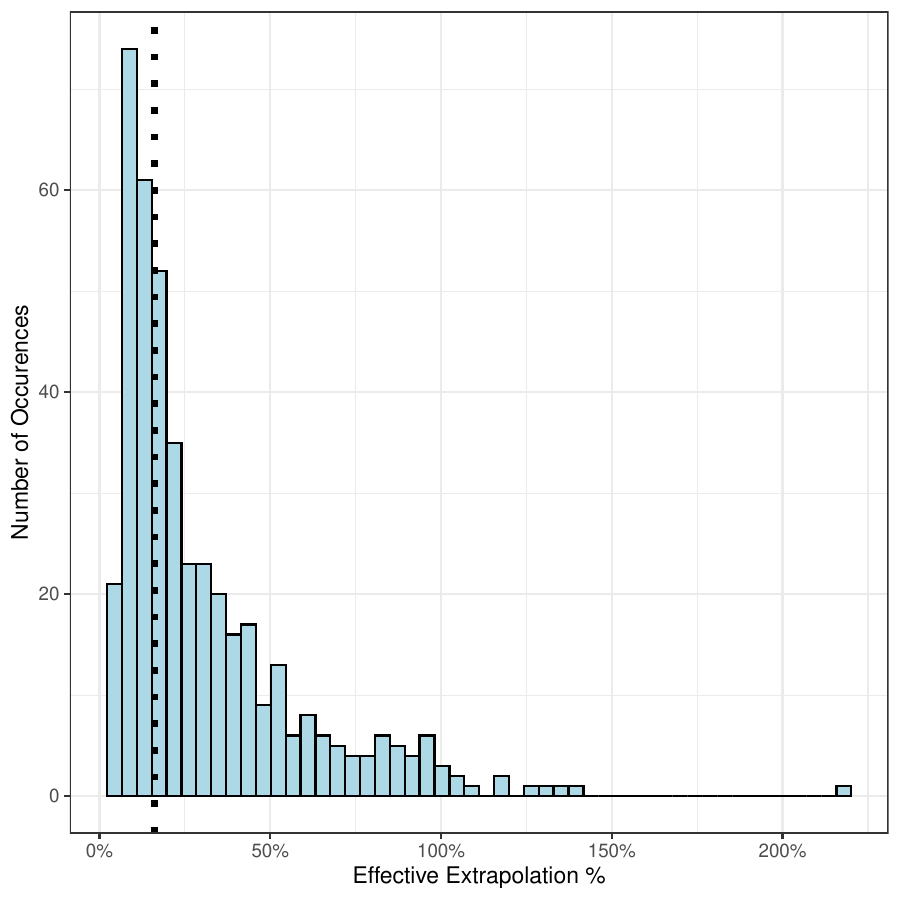}
       \caption{Histogram of the percentage of negative weights by hospital. Dotted line represents extrapolation percentage for entire sample: 15\%}
  \label{fig:hist2}
\end{subfigure}%
\quad\begin{subfigure}[t]{0.48\textwidth}
    \centering
    \includegraphics[width=0.95\textwidth]{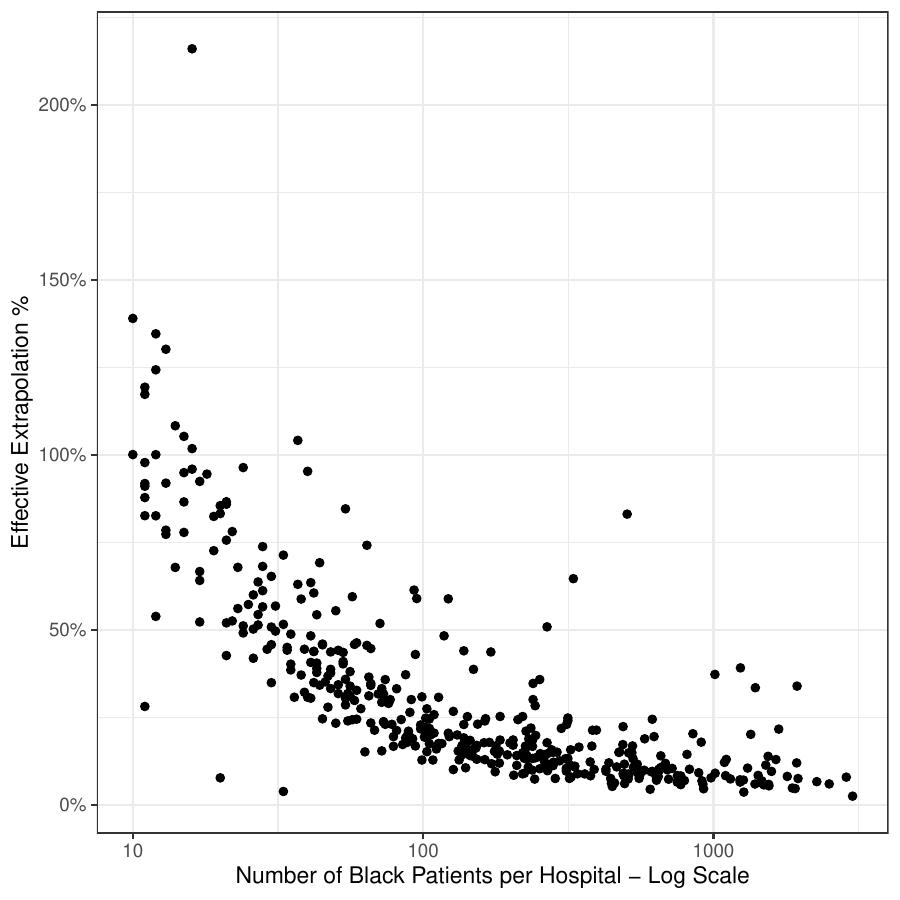}
      \caption{Effective extrapolation percentage plotted against number of Black patients on a log scale.}
  \label{fig:scat}
\end{subfigure}
\end{figure}

To help understand the source of the extrapolation, Figure~\ref{fig:scat} plots the hospital-specific extrapolation percentage against the number of Black patients in each hospital. In this figure, we observe that high levels of extrapolation occur in hospitals with few Black patients. This is not surprising: when there are few Black patients the weights tend to extrapolate to balance the covariate distributions. As such, the restricted version of the weights is useful to prevent extrapolation beyond the support of the data, at the cost of reduced balance. 

\subsection{Disparity estimates with the full adjustment set}

We now present the disparity estimates. We first focus on the full EGS patient population. Table~\ref{tab1} contains disparity estimates for four different scenarios. Without any adjustment, the observed adverse event rate for Black patients is 27\%, and the estimated adverse event rate for white patients is 29.1\%, for a risk difference of -2.1 percentage points and a risk ratio that is 7\% lower. In addition, there is no difference in terms of length of stay.

\begin{table}[htb]
\centering
\begin{threeparttable}
\caption{Estimated racial disparities for outcomes for the full EGS population, with different adjustment sets}
\label{tab1}
\begin{tabular}{lcccc}
\toprule
& Unadjusted & Age- \& sex-adjusted & Fully-adjusted & Fully-adjusted disparity \\
 &            & disparity$^*$ & disparity  & within hospital \\
\midrule
 \multirow{ 2}{2.5cm}{Adverse Event}  & -0.021 & 0.017 & 0.011 & 0.004  \\ 
 & [ -0.024 , -0.018 ] & [ 0.011 , 0.023 ] & [ 0.007 , 0.015 ] & [ -0.003 , 0.011 ]  \\    
\midrule
\multirow{ 2}{2.5cm}{Adverse Event\\(Risk Ratio)}  & 0.93 & 1.07 & 1.04 & 1.01  \\ 
 & [ 0.92 , 0.94 ] & [ 1.05 , 1.09 ] & [ 1.03 , 1.05 ] & [ 1 , 1.03 ]  \\ 
\midrule
\multirow{ 2}{2.5cm}{Length of Stay} & -0.0005 & 0.58 & 0.47 & 0.2  \\ 
  & [ -0.04 , 0.04 ] & [ 0.47 , 0.68 ] & [ 0.43 , 0.51 ] & [ 0.09 , 0.31 ]  \\ 
\bottomrule
\end{tabular}
\begin{tablenotes}[para]
{ \footnotesize Note: Numbers in brackets are 95\% confidence intervals. $^*$ the outcome-allowable covariates $X^y$.} 
\end{tablenotes}
\end{threeparttable}
\end{table}

Following the setup in Section \ref{sec:prelim}, we first consider adjusting for the outcome-allowable covariates --- age and sex --- alone. We now find that the re-weighted adverse event rate for white patients is 25\% for a risk difference of 1.7 percentage points. We also find that Black patients have an average length of stay that is more than half a day longer than white patients. As such, once we compare patients with similar age and sex profiles, we find a clear racial disparity for the EGS population. 

Next, we present results when we include the full set of patient-level covariates. Conditioning on these additional covariates shrinks the estimated disparity by a small amount, but the estimated disparity is similar. We then account for hospitals using the within-hospital weights. In the within-hospital analysis, the estimated adverse event rate for the re-weighted white population is now 26.6\% for a risk difference of 0.4 percentage points; the risk ratio is now 1\% higher. The average length of stay for the re-weighted white population is 5.9 days for a difference of 0.20 days. Overall, accounting for differences across hospitals nearly eliminates the estimated disparities in adverse events and length of stay. These results are consistent with those in \citet{silber2014racial} and demonstrate both the need to ensure that Black and white patients are comparable before estimating disparities, and how the results depend on which covariates are considered allowable. 

Table~\ref{tab2} contains analogous results for the subset of the EGS population that received surgical care. The pattern of results is roughly the same as in the full EGS population. In the unadjusted analysis, we find that Black patients either have very similar or lower rates of adverse events. However, once we adjust for the full set of observed covariates, Black patients have worse outcomes: the risk of an adverse event is 6\% higher for Black patients, and the average length of stay is 0.62 days longer. As before, the disparity shrinks in the within-hospital analysis: the risk of an adverse event is 1\% higher for Black patients, and the average length of stay is 0.25 days longer for Black patients.

\begin{table}[htb]
\centering
\begin{threeparttable}
\caption{Estimated Racial Disparities for Outcomes Following Surgery for an EGS Condition}
\label{tab2}
\begin{tabular}{lcccc}
\toprule
& Unadjusted  & $X\text{-adjusted}$ & $X\text{-adjusted disparity}$ \\
 &             & disparity  & within hospital \\
\midrule
 \multirow{ 2}{3.5cm}{Adverse Event} & -0.002 & 0.02 & 0.01  \\
 & [ -0.01 , 0.001 ] & [ 0.01 , 0.02 ] & [ 0.--1 , 0.02 ]  \\   
\midrule
\multirow{ 2}{3.5cm}{Adverse Event\\(Risk Ratio)} & 0.99 & 1.06 & 1.02  \\  
& [ 0.98 , 1.01 ] & [ 1.04 , 1.07 ] & [ 1 , 1.05 ]  \\
\midrule
\multirow{ 2}{3.5cm}{Length of Stay} & 0.84 & 0.62 & 0.25  \\ 
 & [ 0.75 , 0.94 ] & [ 0.56 , 0.68 ] & [ 0.13 , 0.37 ]  \\  
 \bottomrule
\end{tabular}
\begin{tablenotes}[para]
{ \footnotesize Note: Numbers in brackets are 95\% confidence intervals.} 
\end{tablenotes}
\end{threeparttable}
\end{table}

These results highlight the critical role of the choice of the covariate adjustment set. Without adjustment, we find Black patients have better or comparable outcomes, but this comparison does not reflect important differences in baseline demographics and other attributes.
If we adjust for age and sex, we find clear evidence of racial disparities in EGS. Under the more standard KOB approach that re-weights all available covariates, the estimated disparity remains larger than the unadjusted estimates. Finally, adjustment for hospitals eliminates or substantially reduces the disparity.

In short, Black EGS patients experience worse health outcomes than comparable white EGS patients \emph{across} hospitals. Compared to white patients \emph{within} the same hospital that are comparable on observable features, however, health outcomes for Black patients are not substantially worse. Below, we further consider possible disparities in the surgery decision, which also varies across hospitals.

\subsubsection{Comparison to Matching}
Matching methods are widely used to estimate racial disparities \citep{silber2013characteristics,silber2014racial,silber2015examining,rosenbaum2013using}. 
To better understand differences and similarities between matching and our approach, we re-analyzed the same data set using matching. In particular, we implemented an optimal match with refined-covariate (RC) balance constraints and a propensity score caliper, while exactly matching on hospitals and stratifying the match by surgical status \citep{Pimentel:2015a}. We then calculated disparities using the matched data set. We find that the estimated risk difference for adverse events is 2.5 percentage points, and the average difference in length of stay is 0.51 days. That is, for the matched analysis, we find that the estimated disparities are close to the balancing weight estimates that do not control for hospitals.

One possible explanation for the differences is that while the exact match on hospitals accounts for across-hospital differences, matching leaves large imbalances in other covariates. To that end, we computed the percentage of imbalance reduction across the full set of covariates within each hospital compared to the baseline level of imbalance for matching and weighting. Here, we find that  weighting substantially outperforms matching. For matching, the percentage of bias reduction is 19.7\%; for weighting, the  reduction is 60.2\%. We also calculated the RMSI for matching and plotted it against the RMSI for the restricted weights. The results are in Figure~\ref{fig:match-bal}. We find that the RSMSI for matching is systematically higher than for weighting, which implies that the balance under matching is generally worse. Figure~\ref{fig:match-out} plots the difference between the estimated difference in adverse events for weighting and matching at the hospital level against the RMSI for matching. Here, we find a strong correlation between larger differences in outcomes and larger imbalances after matching. That is, the results based on matching appear to be driven by the fact that matching is unable to balance covariates within hospitals.

\begin{figure}[tbp]
  \centering
  \begin{subfigure}[t]{0.4\textwidth}
    \includegraphics[width=\textwidth]{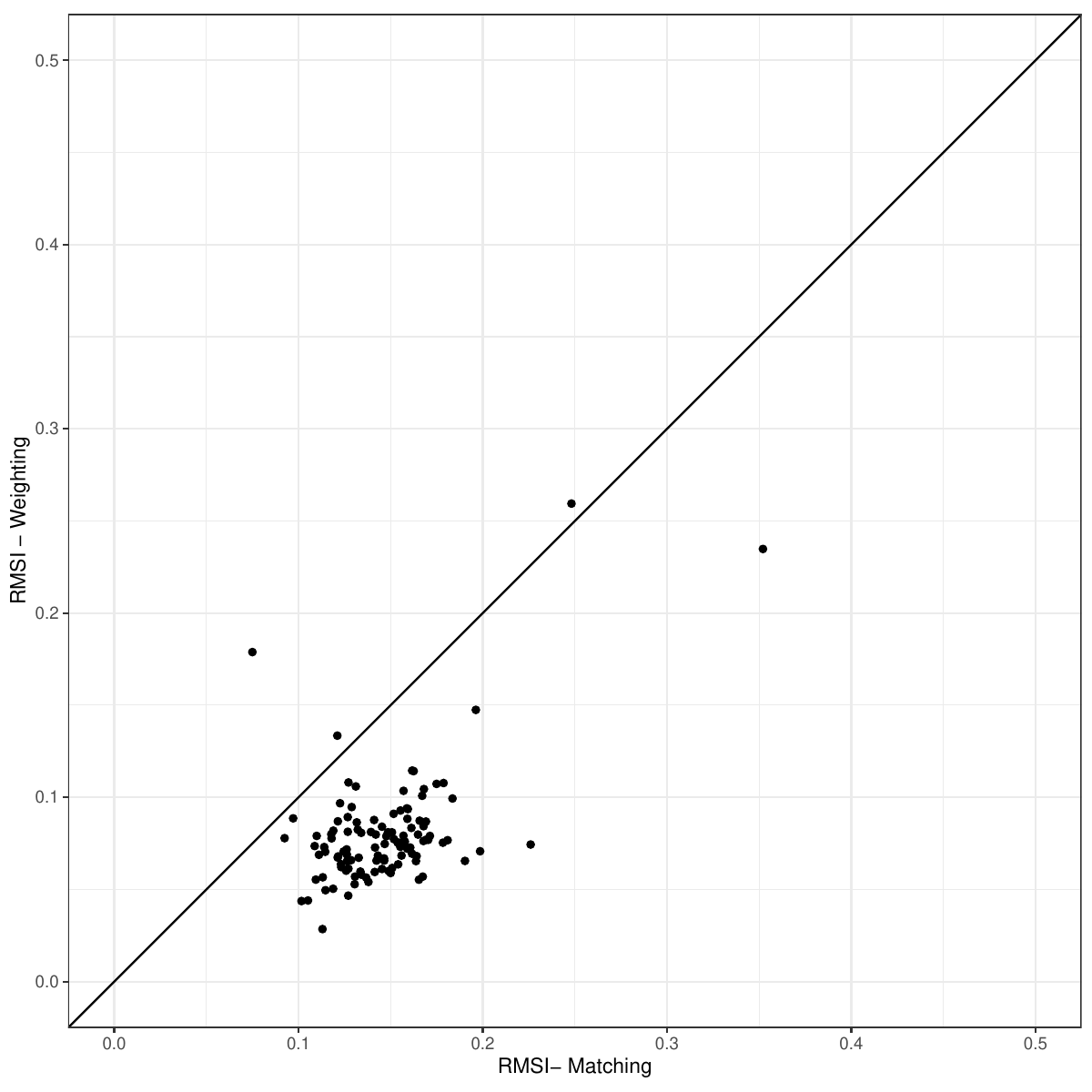}
    \caption{Balance Comparison}
    \label{fig:match-bal}
  \end{subfigure}
  \begin{subfigure}[t]{0.4\textwidth}
    \includegraphics[width=\textwidth]{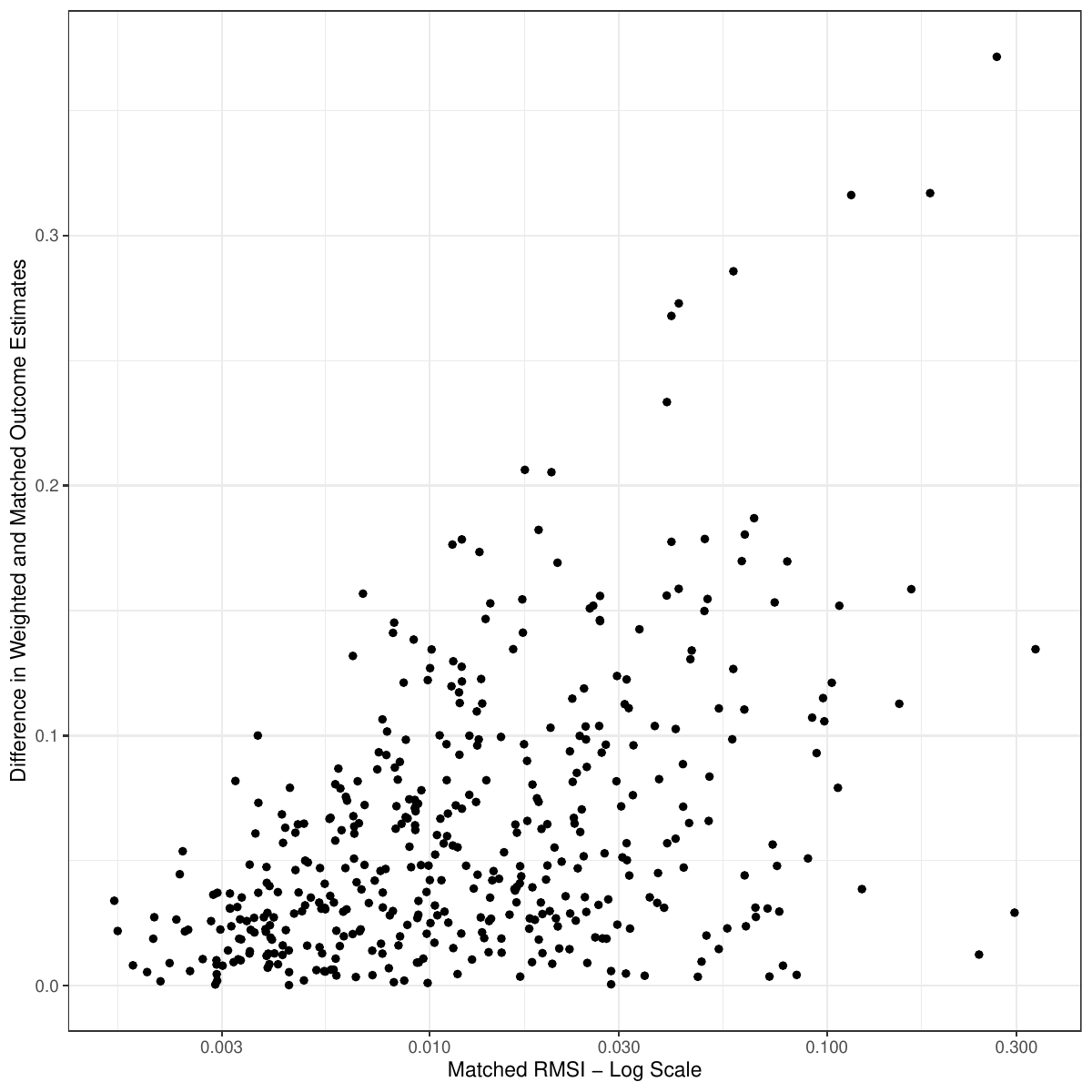}
    \caption{Estimate difference versus balance}
    \label{fig:match-out}
  \end{subfigure}
  \caption{Comparisons between results from restricted weights and results from matching.}
  \label{fig:match-comp}
\end{figure}

Finally, we compared the methods in terms of computation time, using a Macintosh desktop with 64 GB of RAM. Here, we only compared the times for the within-hospital analysis. The RC balance match required 190 minutes. Estimation of the  weights for the within hospital analysis required 7 minutes. We also attempted to include a match with exact matching on states. However, we found these matches took longer than 24 hours and thus did not let them run to completion. Thus, balancing weights require much lower computation times than extant matching methods.

\subsection{Disparity Decompositions}
\label{sec:allowable_decomp}
\label{sec:covs}

Finally, we decompose the estimated disparities into terms that are a function of the decision to treat via surgery, as outlined in Section~\ref{sec:decompose}. We begin the analysis by simply estimating whether Black patients are more or less likely to receive operative care for an EGS condition. A logistic regression of the surgery indicator on the indicator for race and the full set of patient-level covariates shows that the odds that a Black patient receives surgery relative to a white patient with the same covariates are notably lower (0.83, 95\% CI: 0.81, 0.84 ).  

First, we follow the rubric in \citet{jackson2020meaningful}. For this decomposition, we need to partition the covariates into the appropriate sets. We designate the \emph{outcome-allowable} covariates, $X^y$, as age and sex; and designate the \emph{treatment-allowable} covariates, $X^w$, as: the number of comorbidities, sepsis, disability, the 31 comorbidities, and the 51 EGS categories. Finally, we designate indicators for health insurance status and the hospital identifier as the \emph{non-allowable} covariates, $X^n$. Next, we decompose the $X^y$-adjusted disparity into the disparity reduction component and the residual disparity component.  As discussed in Section~\ref{sec:estimation}, we estimate a logistic regression model of the indicator for surgery on the outcome- and treatment-allowable covariates, $X^y$ and $X^w$, in the white patient population. We then use these estimates to generate an estimated target distribution to weight the Black patient population towards to equalize the odds of being treated via surgery. Table~\ref{tab3} contains the results, re-reporting the estimated disparity after adjusting for outcome-allowable covariates (age and sex) only.  

For both outcomes, the estimated disparity reduction from equalizing the odds of surgery between Black and white patients is close to zero, which means that the estimated residual disparity is essentially identical to the observed disparity. This suggests that the disparities we observe are neither exacerbated nor ameliorated by the decision to operate. The source of the disparity does not appear to be a function of racial bias in the decision to operate on EGS patients.

\begin{table}[h]
\centering
\begin{threeparttable}
\caption{Estimated Racial Disparities for Outcomes Following Surgery for an EGS Condition -- Stochastic Intervention Decomposition }
\label{tab3}
\begin{tabular}{lcccc}
\toprule
&  $X^y$ disparity & Disparity reduction & Residual disparity \\
\midrule
 \multirow{ 2}{3.5cm}{Adverse Event}   & 0.017 & -0.0014 & 0.0187  \\ 
 & [ 0.011 , 0.023 ] & [ -0.006 , 0.003 ] & [ 0.013 , 0.024 ]  \\   
\midrule
\multirow{ 2}{3.5cm}{Length of Stay}   & 0.58 & -0.008 & 0.58  \\ 
& [ 0.47 , 0.68 ] & [ -0.08 , 0.06 ] & [ 0.52 , 0.64 ]  \\ 

\bottomrule
\end{tabular}
\begin{tablenotes}[para]
{ \footnotesize Note: Numbers in brackets are 95\% confidence intervals. Estimates in the Table are risk differences. } 
\end{tablenotes}
\end{threeparttable}
\end{table}

Next, we follow the unconditional decomposition of  \citet{Yu2023_decomp}. As we outlined above, we decompose the total disparity into four different terms: base, prevalence, effect, and selection. Notably, this decomposition is unconditional in that we do not adjust for the patient level covariates; they are only used to equalize the odds of surgery within groups and hospitals.\footnote{For this analysis, we had to restrict the study sample to hospitals with at least 15 Black patients; in some small hospitals, all Black patients received surgery. As such, there are slight differences between the unadjusted results presented here and those presented above. Table~\ref{tab4} contains the results.} For both outcomes, we find that the total and base effects are nearly identical, since the estimated prevalence, effect, and selection effects are all small in magnitude. As such, these results mirror the conclusions from  the prior decomposition: the decision to operate and the effects of surgical care do not appear to contribute to racial disparities in EGS outcomes.

\begin{table}[h]
\centering
\begin{threeparttable}
\caption{Estimated Racial Disparities for Outcomes Following Surgery for an EGS Condition -- Unconditional Decomposition }
\label{tab4}
\begin{tabular}{lccccc}
\toprule
&  Total & Base & Prevalence & Effect & Selection \\
\midrule
 \multirow{ 2}{1.5cm}{Adverse Event}  & -0.022 & -0.029 & 0.00019 & 0.006 & 0.00071  \\ 
 & [ -0.023 , -0.021 ] & [ -0.035 , -0.023 ] & [ 0 , 0.00037 ] & [ 0.002 , 0.01 ] & [ -0.003 , 0.004 ]  \\  
\midrule
\multirow{ 2}{1.5cm}{Length of Stay}   & -0.009 & -0.007 & 0.001 & 0.013 & -0.016  \\ 
 & [ -0.025 , 0.006 ] & [ -0.086 , 0.071 ] & [ -0.001 , 0.003 ] & [ -0.039 , 0.066 ] & [ -0.0321 , 0.0001 ]  \\
\bottomrule
\end{tabular}
\begin{tablenotes}[para]
{ \footnotesize Note: Numbers in brackets are 95\% confidence intervals. Estimates in the Table are risk differences. } 
\end{tablenotes}
\end{threeparttable}
\end{table}

\section{Conclusion}
\label{sec:conc}

In the United States, racial disparities are widespread across a broad range of health outcomes. In this paper, we focus on disparities in outcomes for emergency general surgery patients. Analyzing a database of medical claims in three large US states, we find that in the unadjusted data, Black patients do not exhibit worse outcomes than white patients.  However, when we adjusted for either age and sex alone or the full set of patient covariates, we find clear evidence of race-related disparities. Additionally adjusting for admitting hospital largely eliminates disparities in adverse events, though Black patients still have longer post-treatment hospital stays than re-weighted white patients. Using a decomposition analysis, we find that the disparity does not appear to be driven by differences in the decision to operate on Black or white EGS patients.

Consistent with prior studies, we find that patient characteristics explain a large portion of the observed disparities in surgical outcomes between Black and white patients \citep{silber2015examining,tsai2014disparities}. Further, our results suggest that variation across hospitals, including quality of care, is an independent, important driver of disparities in these health outcomes. The finding that hospital quality is a driver of surgical disparities is not new, but further emphasizes that interventions targeted at hospital quality may be critical for further reducing racial disparities in surgical outcomes, and points to larger systemic factors that may be driving disparities. Historically few interventions have addressed variation in hospital quality especially among ``minority-serving'' institutions \citep{khera2015racial,tsai2014disparities,dimick2013black}. Several barriers to addressing the role of hospitals in surgical disparities exist including structural issues related to insurance coverage and resource allocation. 

Methodologically, we offer a unifying framework for linear estimators that adjust for observable differences between groups with the goal of understanding outcome disparities through decompositions \citep{strittmatter2021gender, jackson2022observational, Yu2023_decomp}. In our application, we find that two special cases of linear estimators, regression and matching, perform poorly: regression allows for extreme extrapolation for many hospitals and matching fails to adequately balance patient-level characteristics. The restricted weighting approach, by contrast, gives researchers direct control over the bias-variance trade-off and over extrapolation. 

There are several avenues for future developments. First, how best to balance high-dimensional features remains an open question. Here we use sample splitting and random forests to learn a basis to balance. While reasonable in this setting, several recent papers offer flexible alternatives, including via other machine learning methods \citep{wang2021flame}. Second, we have focused exclusively on linear weighting estimators. A natural extension is to instead consider so-called \emph{augmented} weighting estimators, which combine outcome modeling and weighting \citep[see][]{benmichael2021_review}. These approaches can adjust for imbalances that remain after weighting or matching alone, albeit at the cost of additional modeling assumptions and possible extrapolation. Finally, our analysis is in the spirit of Kitagawa-Oaxaca-Blinder decompositions and other adjustment methods that adjust for observable baseline differences between groups. Even setting aside fundamental questions about defining racial disparities \citep{jackson2022observational}, these estimated disparities might reflect important differences in which patients are admitted to different hospitals \citep{hull2018estimating}. While patients will have less choice in admitting hospital for emergency general surgery than for non-emergency general surgery, properly accounting for (observed or unobserved) differences in what types of patients are admitted to different hospitals is an important next step in further understanding health disparities. 

\clearpage

\bibliographystyle{chicago}
\bibliography{references}

\end{document}


\noindent%
	{\bfseries\large
Supplement to ``Measuring Racial Disparities in Emergency General Surgery''}

\noindent%
	
	\noindent%

\pagenumbering{arabic}
\setcounter{equation}{0}
\setcounter{table}{0}
\setcounter{figure}{0}
\renewcommand{\theequation}{S\arabic{equation}}
\renewcommand{\thetable}{S\arabic{table}}
\renewcommand{\thefigure}{S\arabic{figure}}

\section{Proofs and derivations}
\label{sec:proofs}

Key proofs and derivations from the paper are included below.

\begin{proof}[Proof of Proposition~\ref{prop:ob_mse}]

First note that the intercept $\hat{\alpha}_g$ and the coefficients $\hat{\beta}_g$ are given by
\[
\hat{\alpha}_g = \bar{Y}_g \;\;\; \hat{\beta}_g = \Sigma_g^{-1}\sum_{G_i = g} X_i (Y_i - \bar{Y}_g),
\]
where $\bar{Y}_g = \frac{1}{n_g}\sum_{G_i = g}Y_i$. Now consider the noiseless least squares problem
\[
\min_{\alpha, \beta} \frac{1}{n_g}\sum_{G_i = g}\left(m(X_i, g) - \alpha - \beta \cdot X_i\right)^2,
\]
with solutions
\[
\tilde{\alpha}_g = \bar{m}_g \;\;\; \tilde{\beta}_g = \Sigma_g^{-1} \rho_g.
\]

Now we can decompose the estimation error into
\[
\begin{aligned}
  \hat{\mu}_g^{\mathcal{P}^\ast OB} - \mu_g^{\mathcal{P}^\ast} & = \hat{\alpha}_g - \tilde{\alpha}_g + \left(\hat{\beta}_g - \tilde{\beta}_g \right) \cdot \E_{\mathcal{P}^\ast}[X] + \tilde{\alpha}_g + \tilde{\beta}_g  \cdot \E_{\mathcal{P}^\ast}[X] - \E_{\mathcal{P}^\ast}[m(X, g)]\\
  & = \bar{Y}_g - \bar{m}_g + \Sigma_g^{-1} \sum_{G_i = g} X_i \left(Y_i - \bar{Y}_g - m(X_i, g) - \bar{m}_g\right) \cdot \E_{\mathcal{P}^\ast}[X] + \bar{m}_g + \Sigma_g^{-1}\rho_g - \E_{\mathcal{P}^\ast}[m(X, g)].
\end{aligned}
\]
Now the conditional bias is
\[
 \E\left[\hat{\mu}_g^{\mathcal{P}^\ast OB} - \mu_g^{\mathcal{P}^\ast} \mid X, G\right] = \Sigma_g^{-1}\rho_g - \left(\E_{\mathcal{P}^\ast}[m(X, g)] - \bar{m}_g \right),
\]
and the conditional variance is
\[
\Var\left(\hat{\mu}_g^{\mathcal{P}^\ast OB}  \mid X, G\right) = \frac{\sigma^2}{n_g} + \frac{1}{n_g^2}\sum_{G_i = g}\left(\E_{\mathcal{P}^\ast}[X] \Sigma_g^{-1}X_i\right)^2 \sigma^2.
\]
So the conditional MSE is the sum of the squared conditional bias and the conditional variance.

\end{proof}

\subsection{Equivalence of penalized linear weighting estimator and regularized OB estimator.}

\citet{benmichael2021_review} derive the general Lagrangian dual for optimization problems of the form in Equation~\ref{eq:opt_prob}. Specializing those results, the Lagrangian dual solves the unconstrained optimization problem:

\[
\min_{\eta} \frac{1}{2} \sum_{j=1}^J \left[\sum_{G_i = g, H_i = j}\left(\eta_{0j} + \phi(Z_i) \cdot \eta_j\right)^2 - P^\ast(H = j) \eta_{0j} + \E_{\mathcal{P}^\ast}[\phi(Z) \mid H = j] \cdot \eta_j \right] + \frac{\lambda}{2}\sum_j \|\eta_j - \bar{\eta}\|_2^2,
\]
where $\bar{\eta} = \frac{1}{J}\sum_{j=1}^J \eta_j$, and the corresponding weights are
\[
\gamma(X_i) = \sum_{j = 1}^J \mathbbm{1}\{H_i = j\}\left(\hat{\eta}_{0j} + \phi(Z_i) \cdot \hat{\eta}_j\right).
\]
To show the equivalence, we will work with this in matrix form. Define the matrix $\Phi \in \R^{n \times J(p+1)}$ with elements $\Phi_{i,jk} = \bbone\{G_i = g, H_i = j\}\phi_k(Z_i)$, where $\phi_0(z) = 1$; the matrix $D \in \R^{J(p+1) \times J(p+1)}$ with elements $D_{jk, j'k'} = \bbone\{j = j', k = k', k \neq 0\}$; and the matrix $A \in \R^ {J(p+1) \times J(p+1)}$ with elements $A_{jk, j'k'} = \frac{1}{J}\bbone\{k = k', k \neq 0\}$. Also define the vector $\Phi^\ast \in \R^{J(p+1)}$ with elements $\phi^\ast_{jk} = P^\ast(H = j)\E_{\mathcal{P}^\ast}[\phi_k(Z) \mid H = j]$ and $\tilde{Y} \in \R^{n}$ where $\tilde{Y}_i = \bbone\{G_i = g\}Y_i$. Then the dual can be written as
\[
\min_\eta \; \frac{1}{2}\|\Phi \eta\|_2^2 - \phi^\ast \cdot \eta + \frac{\lambda}{2} \|(D - A)\eta\|_2^2.
\]
The solution is
\[
\hat{\eta} = \left(\Phi'\Phi + \lambda (D - A)\right)^{-1}\phi^\ast,
\]
and so the linear weighting estimator is
\[
\hat{\mu}_g^{\mathcal{P}^\ast} = \hat{\gamma}'\tilde{Y} =  {\phi^\ast}'\left(\Phi'\Phi + \lambda (D - A)\right)^{-1}\Phi'\tilde{Y}.
\]
Now compare to Equation \eqref{eq:ridge_model}, reproduced here in matrix form as
\[
\min_\beta \; \|\tilde{Y} - \Phi \beta\|_2^2 + \lambda \|(D - A)\eta\|_2^2,
\]
where we have used the fact that the solution $ \hat{\mu}_{\beta_g} = \frac{1}{J}\sum_{j=1}^J \hat{\beta}_{gj}$.
The solution is
\[
\hat{\beta} = \left(\Phi'\Phi + \lambda (D - A)\right)^{-1}\Phi'\tilde{Y},
\]
and the corresponding OB estimate is
\[
\hat{\mu}_g^{\mathcal{P} \ast} = {\phi^\ast}'\hat{\beta} =  {\phi^\ast}'\left(\Phi'\Phi + \lambda (D - A)\right)^{-1}\Phi'\tilde{Y},
\]
which is equivalent to the weighting estimate.

\section{Simulation Study}
\label{sec:sim}

Here, we conduct two simulation studies. We focus on two key questions. In the first study, we seek to understand how the performance of the balancing weights changes with different hyperparameter values. In the second study, we examine the performance of standard errors based on RVE relative to HC2 standard errors.
In our DGP, we use five observed covariates that are independent draws from a multivariate standard Gaussian distribution: $X = (X_1, \dots, X_5) ^\top$. We construct a group indicator $G$ as $G = \mathbb{I}(G^* > 0)$ where $G^* = (1.5 X_1 + 1.5 X_2 + .7 X_1 X_2)/c + \text{Unif}(-0.5, 0.5)$. Next, we generate outcomes for units as:
\[
Y = 5G + X_2 + X_3 + \epsilon
\]
and $\epsilon$ is drawn from $N(0,1)$. The covariates $X_4$ and $X_5$ are additional covariates that differ between the two groups but are not related to the outcome . We use $c$ to control the overlap between the two group distributions, using values of $c = 1$ (poor overlap) and $c = 10$ (good overlap). 

In the simulations, we use three different methods of estimation for the weights. First, we use  weights from an unrestricted optimization problem where $L = -\infty$ and $U = \infty$. We refer to this set of weights as ``unrestricted''  weights, since as we noted above, these weights allow for extrapolation. Next, we constrain the weights to be non-negative (i.e. $L = 0, U = \infty$). We refer to this set of weights as ``restricted'' weights. In the simulations below, there are no group-level indicators, and so we remove the global balance constraint in Equation \eqref{eq:opt_prob} and only require approximate balance, with $\lambda$ controlling the bias-variance trade-off as before. Finally, we also implement standard IPW weights fit via logistic regression. For this DGP, the IPW estimator will be consistent when overlap holds. We used a sample size of 1000 and repeated the simulations 1000 times for each condition.

\subsection{Simulation Study 1}
\label{sec:hyper}

In the first simulation study, we examine how varying the hyper-parameter affects the results from a design-based perspective--that is, without considering outcomes. Specifically, we focus on bias as measured by balance on baseline covariates and the effective sample size. As such, in this simulation study, we focus on two different performance metrics. The first metric is the effective sample size defined as: 
\[ 
n^{\text{eff}} \;\; \equiv \;\; \left(\sum_{i=1}^n \hat{\gamma}_i\right)^2 \Bigg/ \sum_{i=1}^n \hat{\gamma}_i^2
 \]
The second metric is the percentage of bias reduction (PBR). PBR measures how much weighting reduces imbalance compared to the unweighted data. We use a PBR metric based on standardized differences. The standardized difference is a common statistic employed to measure imbalance before and after matching or weighting. For covariate $k$, the standardized difference before weighting is
\[
\hat{\Delta}_{k} = \frac{\bar{X}_{1k} - \bar{X}_{0k}}{\sqrt{(V(X_{1k}) + V(X_{0c}))/2}}
\]
\noindent where $X_{1k}$ and $X_{0k}$ are the $G=1$ and $G=0$ group vectors for covariate $k$. The standardized difference after weighting is
\[
\hat{\Delta}_{wk} = \frac{\bar{X}_{w1k} - \bar{X}_{w0k}}{\sqrt{(V(X_{1k}) + V(X_{0c}))/2}}
\]
\noindent where $\bar{X}_{w1k}$ and $\bar{X}_{w0k}$ are weighted means based on the estimated balancing weights. We use $\hat{\Delta}_{uw}$ to refer to the vector of standardized differences for all $K$ covariates in the unweighted data, and $\hat{\Delta}_{w}$ for the vector of standardized differences for all $K$ covariates in the weighted data. Using these two quantities, we then calculate PBR as:
\[ 
PBR = 100\% \times \left[ \frac{1}{K} \sum_k |\hat{\Delta}_{w}| \; \Big/ \; \frac{1}{K} \sum_k |\hat{\Delta}_{uw}| \right]  .
\]
This measure describes the reduction in bias based on the change in balance across all covariates due to weighting. In this study, we are interested in how these metrics change with the hyper-parameter. As such, for both sets of weights, we vary $\lambda$ from 0 to 150 in increments of 15. Small values of $\lambda$ should increase PBR but lower the effective sample size relative to IPW which will have fixed results.

Next, we review the results from the simulations. Figure~\ref{fig:sim1} contains the simulation results in the good overlap scenario. Here, the IPW estimator performs well as it reduces bias by just over 90\% with an effective sample size over under 800. For the restricted and unrestricted weights, the performance depends on primarily on the hyper-parameter values. When $\lambda=0$, we achieve 100\% bias reduction for both restricted and unrestricted  weights, but we pay a price in terms of effective sample size, especially for the restricted weights. This is to be expected given that when we set $\lambda=0$, we give no priority to the sample size in the objective function. For larger values of $\lambda$, the unrestricted  weights do slightly better in terms of bias reduction but produce a similar effective sample size. For mid-range values of $\lambda$, bias reduction and effective sample size are very similar to the IPW weights. When $\lambda$ is large, the effective sample size is larger, but we pay a price in terms of PBR, which drops below 80\%. In general, though, when overlap is good, one can use lower $\lambda$ values to beat the performance of the IPW estimator. 

\begin{figure}
  \centering
    \includegraphics[scale=0.6]{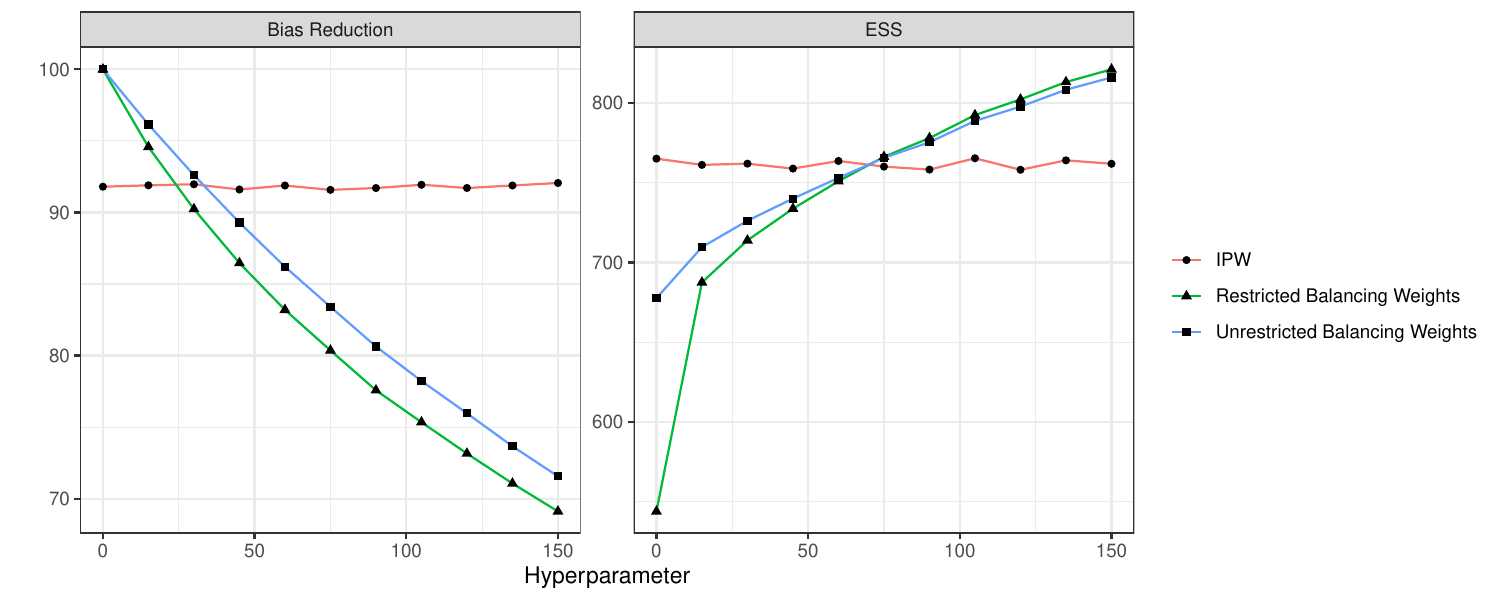}
    \caption{Bias reduction and effective samples across all simulation runs -- good overlap scenario. $X$-axis is $\lambda$ which controls the bias variance tradeoff for the balancing weights.}
  \label{fig:sim1}
\end{figure}

Figure~\ref{fig:sim2} shows the results from the poor overlap scenario. With poor overlap, IPW is no longer effective at balancing the covariates as the PBR hovers around zero. Here, the performance of the IPW estimator mirrors that found in \citet{Kang:2007}. That is, the IPW estimator essentially fails when overlap is poor. Both restricted and unrestricted weights still lead to substantial bias reduction in this scenario: unrestricted weights achieve exact balance (100\% bias reduction) while restricted weights reduce estimated bias by 80\%. Thus, when overlap is poor, allowing extrapolation in the weights allows for further bias reduction --- albeit at the price of some additional model dependence. For the restricted weights, the effective sample size here is quite small when $\lambda = 0$ but a slightly higher value of $\lambda$ produces a much larger effective sample size. The unrestricted weights also allow for a much larger effective sample size when $\lambda = 0$. In general, as $\lambda$ increases, the sample size grows and bias reduction levels off at around 50\%. 

\begin{figure}
  \centering
    \includegraphics[scale=0.6]{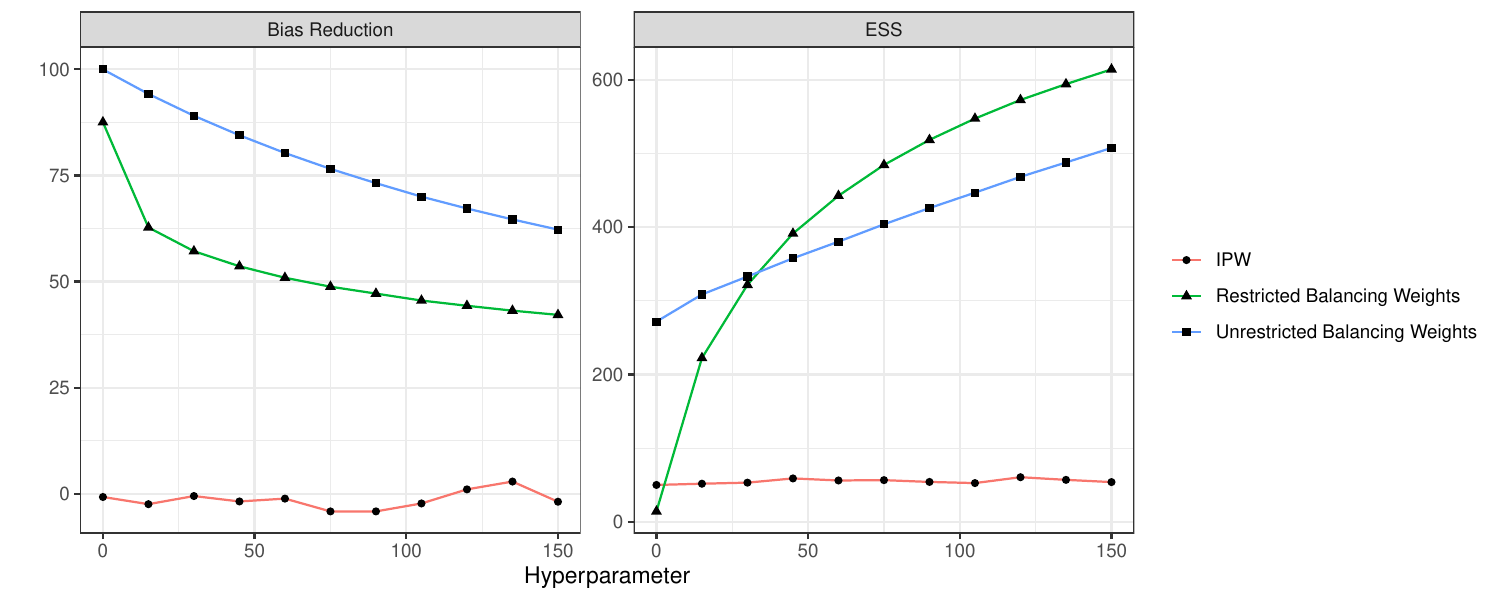}
    \caption{Bias reduction and effective samples across all simulation runs -- poor overlap scenario. $X$-axis is $\lambda$ which controls the bias variance tradeoff for the balancing weights.}
  \label{fig:sim2}
\end{figure}

These simulations reveal two clear patterns. For any weighting estimator, there is a clear bias-variance trade-off. For both overlap conditions, the hyper-parameter controls this bias-variance tradeoff to a much greater extent than the constraints we placed on the weights. When overlap is good, either optimization-based weighting approach will match the performance of the IPW estimator for the some value of $\lambda$, but provide greater flexibility. For example, when the overall sample is large, preserving sample size may be a secondary concern. In these scenarios, the analyst can opt to further reduce bias. Next, we observe how the performance of the IPW estimator is strongly dependent on the overlap condition.  When overlap is poor, the IPW estimator fails. However, for both restricted and unrestricted weights the bias-variance relationship holds, and the analyst can still achieve good performance in terms of both PBR and effective sample size.

\subsection{Simulation Study 2}
\label{sec:coverage_sim}

In the second simulation study, we focus on two issues related to statistical inference. First, we compare relative performance of standard errors based on RVE compared to HC2 standard errors. As we noted above, for the IPW estimator with ATT-like estimands, the standard errors based on weighting alone can be either too large or too small \citet{reifeis2022variance}. Second, we also study whether that is also true for either of the variance estimation methods we outlined for balancing weights. For this simulation, we measured bias, the average standard error estimate across the two methods, and the coverage rate for the 95\% confidence interval. In this simulation study, we again vary the overlap parameter and $\lambda$. Here, we vary $\lambda$ from 0 to 15 in increments of 2.5, and we set $c$ to either 1, 5, or 10. We only use restricted weights for this simulation, since the unrestricted weights will always have larger effective sample sizes than the restricted weights.

\begin{figure}[h]
  \centering
    \includegraphics[scale=0.75]{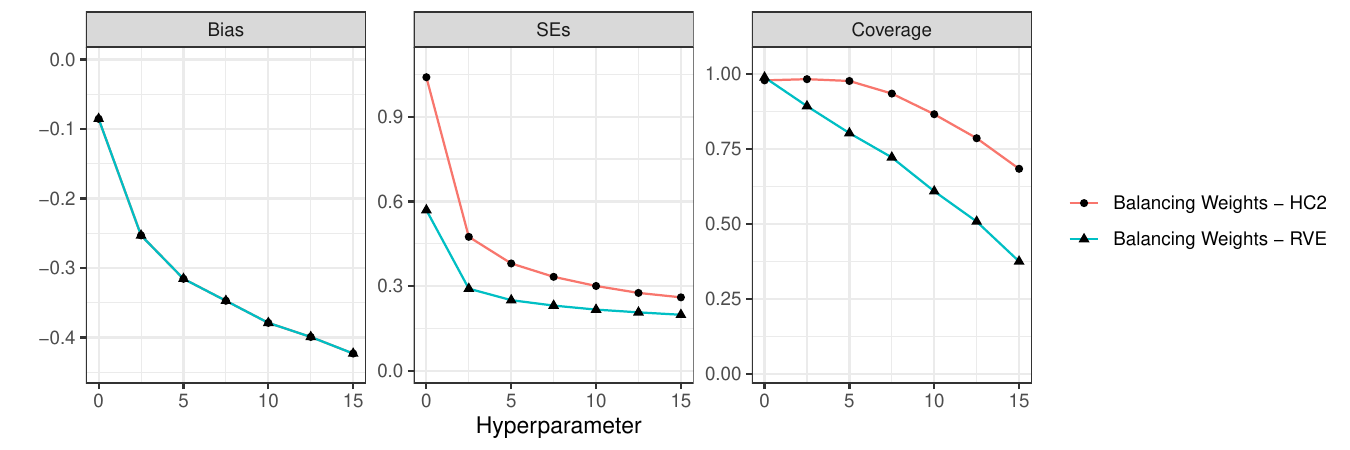}
    \caption{Standard error size and coverage across all simulation runs -- poor overlap scenario. $X$-axis is $\lambda$ which controls the bias variance tradeoff for the balancing weights.}
  \label{fig:se-bad-overlap}
\end{figure}

Figure~\ref{fig:se-bad-overlap} contains the simulation results in the poor overlap scenario. Here we find that the residual variance estimator produces smaller standard error estimates for all values of the hyperparameter. The differences, however, are quite small for larger values of $\lambda$. However, for both methods, we now find that the value of hyperparameter affects the coverage rate: for larger values of $\lambda$, the higher levels of bias result in under coverage. Notably, the larger SEs from using HC2 perform better in that the confidence interval doesn't undercover unless the bias is larger compared to the RVE SEs. In sum, when overlap is good or moderate, the simulations confirm the the estimated standard errors are valid, albeit conservative. When overlap is poor, the estimated confidence intervals can under-cover when the optimization problem doesn't sufficiently prioritize bias. As such, when overlap is poor analysts should generally select small non-zero values of $\lambda$ and use HC2 standard errors.

Next, we report the results from the good and medium overlap scenarios. Figure~\ref{fig:se-good-overlap} contains the simulation results in the good overlap scenario. Here, we observe that the residual variance estimator produces uniformly larger standard errors than the HC2 standard errors. This holds true for any effective sample sizes. However, the difference in standard error magnitude results in very similar coverage rates: both methods over cover the nominal 95\% coverage rate. This results in conservative inferences that are typical for weighting estimators. Specifically, both are conservative in these scenarios, with coverage above 95\% for all values of the hyperparameter, which is typical for weighting estimators. As such, in this scenario, the concerns that \citet{reifeis2022variance} observe with IPW estimators for the ATT do not arise.

\begin{figure}[h]
  \centering
    \includegraphics[scale=0.7]{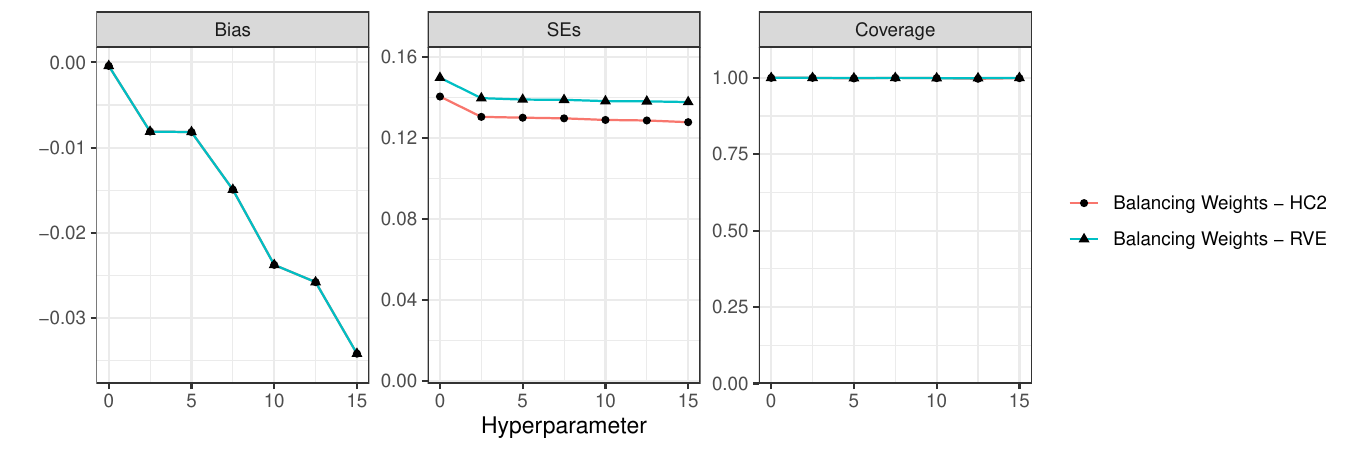}
    \caption{Standard error size and coverage across all simulation runs -- good overlap scenario. $X$-axis is $\lambda$ which controls the bias variance tradeoff for the balancing weights.}
  \label{fig:se-good-overlap}
\end{figure}
 
Figure~\ref{fig:se-med-overlap} contains the results from the medium overlap scenario. In the medium overlap scenario, however, the HC2 SEs are slightly larger across the differing values of $\lambda$. However, again we find that both SEs produce confidence errors that are conservative and over cover.
 
\begin{figure}[h]
  \centering
    \includegraphics[scale=0.7]{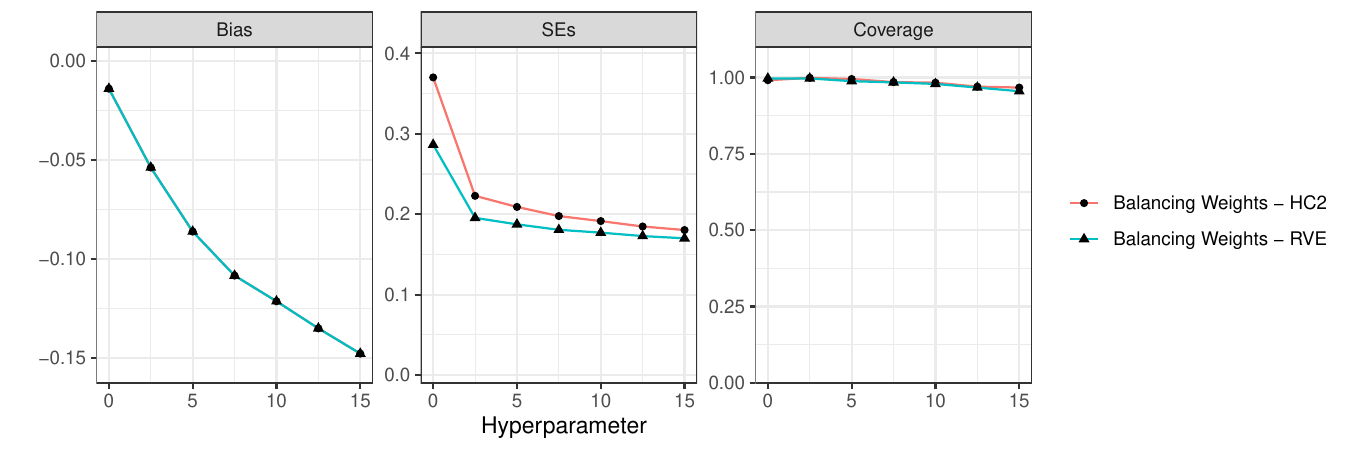}
    \caption{Standard error size and coverage across all simulation runs -- medium overlap scenario. $X$-axis is $\lambda$ which controls the bias variance tradeoff for the balancing weights.}
  \label{fig:se-med-overlap}
\end{figure}

\section{Additional Empirical Results}

\subsection{Full List of Covariates}

\begin{longtable}{l}
\caption{List of Baseline Covariates}\\
\label{tab:covs}
Covariate Name \\
  \hline
Age Spline 1  \\ 
  Age Spline 2  \\ 
  Age Spline 3  \\ 
  Age Spline 4  \\ 
  Age Spline 5  \\ 
  Age Spline 6  \\ 
  No. Comorbidities  \\ 
  Female  \\ 
  Hispanic  \\ 
  Sepsis  \\ 
  Disability  \\ 
  Under 65  \\ 
  Medicare  \\ 
  Medicaid  \\ 
  Private Insurance  \\ 
  Self Insurance  \\ 
  Other  \\ 
  \midrule
  Comorbidities \\
  \midrule
  Congestive Heart Failure  \\ 
  Cardiac Arrhythmias  \\ 
  Valvular Disease  \\ 
  Pulmonary Circulation Disorders  \\ 
  Peripheral Vascular Disorders  \\ 
  Hypertension, Uncomplicated  \\ 
  Paralysis  \\ 
  Other Neurological Disorders  \\ 
  Chronic Pulmonary Disease  \\ 
  Diabetes, Uncomplicated  \\ 
  Diabetes, Complicated  \\ 
  Hypothyroidism  \\ 
  Renal Failure  \\ 
  Liver Disease  \\ 
  Peptic Ulcer Disease Excluding Bleeding  \\ 
  AIDS/HIV  \\ 
  Lymphoma  \\ 
  Metastatic Cancer  \\ 
  Solid Tumor Without Metastasis  \\ 
  Rheumatoid Arthritis/Collagen Vascular  \\ 
  Coagulopathy  \\ 
  Obesity  \\ 
  Weight Loss  \\ 
  Fluid and Electrolyte Disorders  \\ 
  Blood Loss Anemia  \\ 
  Deficiency Anemia  \\ 
  Alcohol Abuse  \\ 
  Drug Abuse  \\ 
  Psychoses  \\ 
  Depression  \\ 
  Hypertension, Complicated  \\ 
  \midrule
  EGS Conditions \\
  \midrule
  Colorectal: Bleeding  \\ 
  Colorectal: Cancer  \\ 
  Colorectal: Colitis  \\ 
  Colorectal: Diverticulitis  \\ 
  Colorectal: Fistula  \\ 
  Colorectal:Hemorrhage  \\ 
  Colorectal: Megacolon  \\ 
  Colorectal:Colostomy/Ileostomy  \\ 
  Colorectal:Rectal Prolapse  \\ 
  Colorectal: Regional Enteritis  \\ 
  Colorectal: Anorectal Stenosis, Polyp, Ulcer, NEC  \\ 
  Gen Abdominal: RP Abscess  \\ 
  Gen Abdominal: Hemoperitoneum  \\ 
  Gen Abdominal: Abdominal Mass  \\ 
  Gen Abdominal: Abdominal Pain  \\ 
  Gen Abdominal:
               Peritonitis  \\ 
  HPB: Gallstones and Related Diseases  \\ 
  HPB: Hepatic  \\ 
  HPB: Pancreatitis  \\ 
  Hernia: Diaphragmatic  \\ 
  Hernia:Femoral  \\ 
  Hernia: Incisional  \\ 
  Hernia:Inguinal  \\ 
  Hernia: Other  \\ 
  Hernia:Umbilical  \\ 
  Hernia:Ventral  \\ 
  Intestinal Obstruction: Adhesions  \\ 
  Intestinal Obstruction:Incarcerated Hernias  \\ 
  Intestinal Obstruction: Intussusceptions  \\ 
  Intestinal Obstruction: Obstruction  \\ 
  Intestinal Obstruction: Volvulus  \\ 
  Resuscitation: Acute Respiratory Failure  \\ 
  Resuscitation: Shock  \\ 
  Skin \& Soft Tissue: Abscesses  \\ 
  Skin \& Soft Tissue: Cellulitis  \\ 
  Skin \& Soft Tissue: Compartment Syndrome  \\ 
  Skin \& Soft Tissue: Fasciitis  \\ 
  Skin \& Soft Tissue: Pressure Ulcers  \\ 
  Skin \& Soft Tissue: Wound Care  \\ 
  Upper GI: Appendix  \\ 
  Upper GI: Bleed  \\ 
  Upper GI: Bowel Perforation  \\ 
  Upper GI: Fistula  \\ 
  Upper GI: Gastrostomy  \\ 
  Upper GI: Ileus  \\ 
  Upper GI: Meckel's  \\ 
  Upper GI: Peptic Ulcer Disease  \\ 
  Upper GI: Small Intestinal Cancers  \\ 
  Vascular Acute Intestinal Ischemia  \\ 
  Vascular: Acute Peripheral Ischemia  \\ 
  Vascular: Phlebitis  \\ 
   \bottomrule
\end{longtable}

\subsection{Hyperparameter Selection}
\label{sec:hyperparameter}

To better understand the role of $\lambda$, we conducted an additional analysis. In this analysis, we estimated weights for a series of $\lambda$ values. For each $\lambda$ value, we calculated the PBR measure outlined in Section~\ref{sec:sim} and the effective sample size. In this analysis, we use a within-hospital specific measure of PBR. Specifically, we calculate the vectors of standardized differences within each hospital, which we denote as $\hat{\Delta}_{huw}$ and $\hat{\Delta}_{hw}$. For $K$ covariates with $H$ hospitals, the dimension of the vectors is now $K \times H = J$. We can then calculate the reduction in bias for all covariates within all hospitals as:
\[ 
PBR_H = 100\% \times \left[ \frac{1}{J} \sum_J |\hat{\Delta}_{hw}| \; \Big/ \; \frac{1}{H} \sum_j |\hat{\Delta}_{huw}| \right].
\]
\noindent We plot $PBR_H$ against the effective sample size for the overall EGS population and the population that received surgery in Figure~\ref{fig:lambda-hosp}. In both plots, we observe that when $\lambda$ is zero we maximize the amount of bias reduction. Increasing $\lambda$ by a small amount decreases the bias reduction a modest amount but dramatically increases the effective sample size. Overall, we observe the clear bias-variance trade off that is controlled by the value of $\lambda$.

\begin{figure}[h]
  \centering
  \begin{subfigure}[t]{0.45\textwidth}
    \includegraphics[width=\textwidth]{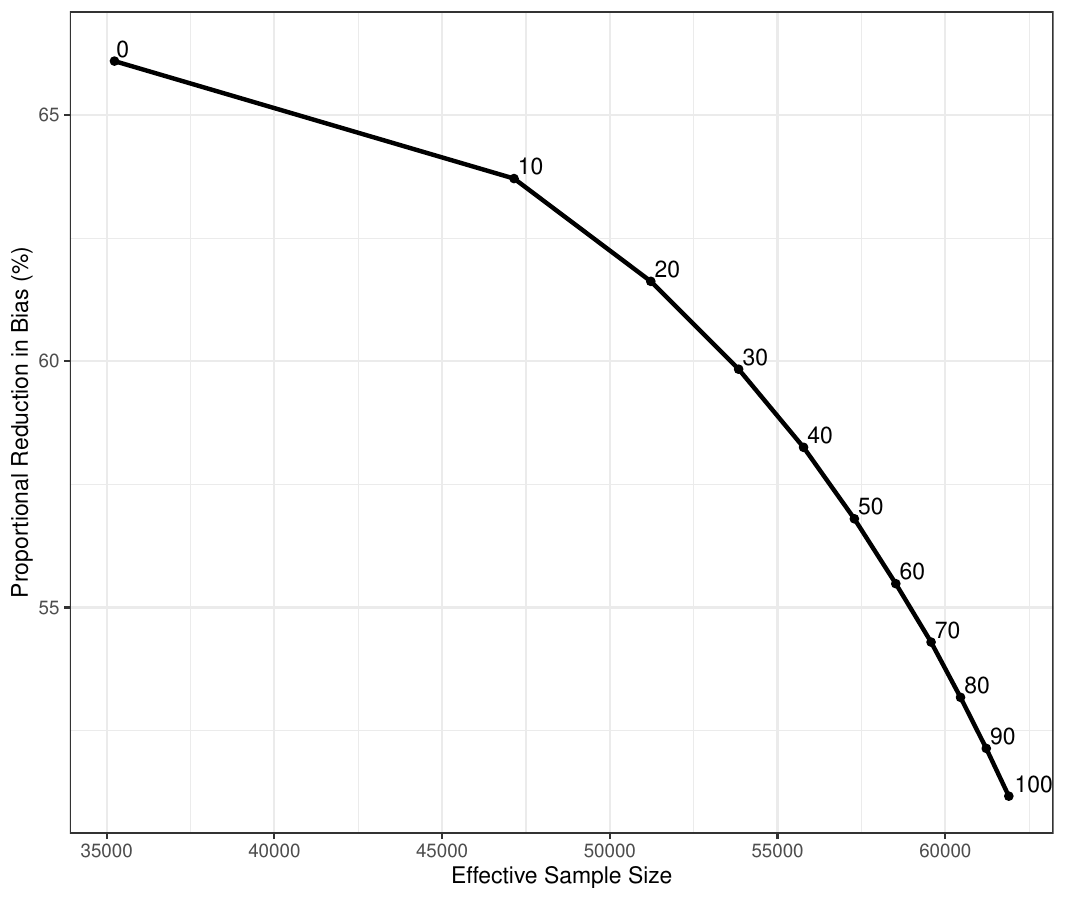}
    \caption{Surgery Population}
    \label{fig:state3}
  \end{subfigure}
  \begin{subfigure}[t]{0.45\textwidth}
    \includegraphics[width=\textwidth]{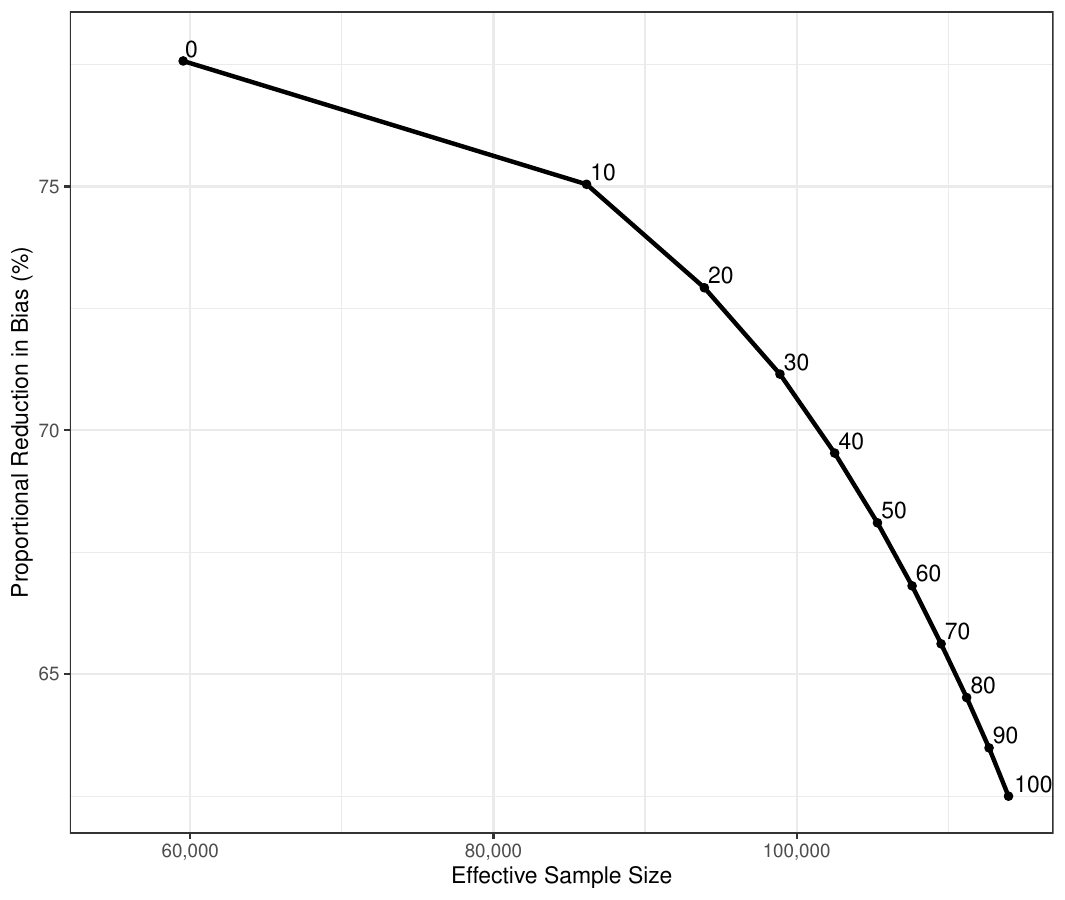}
    \caption{EGS Population}
    \label{fig:state4}
  \end{subfigure}
  \caption{Estimated Bias-Variance Tradeoff as a Function of $\lambda$ values. Each dot represents the estimated Percent Bias Reduction and average effective size for balancing weights with different values of $\lambda$.}
  \label{fig:lambda-hosp}
\end{figure}

\subsection{Balance Diagnostics}

Next, we report on balance statistics. First, we report on the imbalance in the baseline covariates before weighting. In Table~\ref{tab:bal}, we report balance statistics for the subset of covariates with the largest imbalances. In Figure~\ref{fig:bal.plot1}, we plot the balance results for the analysis that uses the full set of covariates but does not account for hospitals. Figure~\ref{fig:bal.plot2} includes the results for the within hospital analysis for those patients that received operative care. Finally, Figure~\ref{fig:bal.plot2} contains the results for $X^y$ subset of covariates.

\begin{table}[ht]
\centering
\caption{Balance Table for EGS Disparity Data: Selected Covariates with Largest Imbalances at Baseline}
\label{tab:bal}
\begin{tabular}{lcccc}
  \toprule
 & Whites & Blacks & Std. Diff. \\ 
  \midrule
Age & 62.5 & 55.9 & -0.29 \\ 
  Under 65 & 50.4 & 66.4 & 0.27 \\ 
  Medicare & 53.5 & 42.4 & -0.18 \\ 
  Medicaid & 8.3 & 20.8 & 0.28 \\ 
  Cardiac Arrhythmias & 25.5 & 17.8 & -0.16 \\ 
  Hypothyroidism & 13.4 & 5.3 & -0.25 \\ 
  Renal Failure & 14.9 & 22.4 & 0.15 \\ 
  Depression & 13.4 & 7.6 & -0.16 \\ 
  Hypertension, Complicated & 13.6 & 21.3 & 0.16 \\ 
 \bottomrule
\end{tabular}
\end{table}

\begin{figure}[htbp]
  \centering
    \includegraphics[scale=0.6]{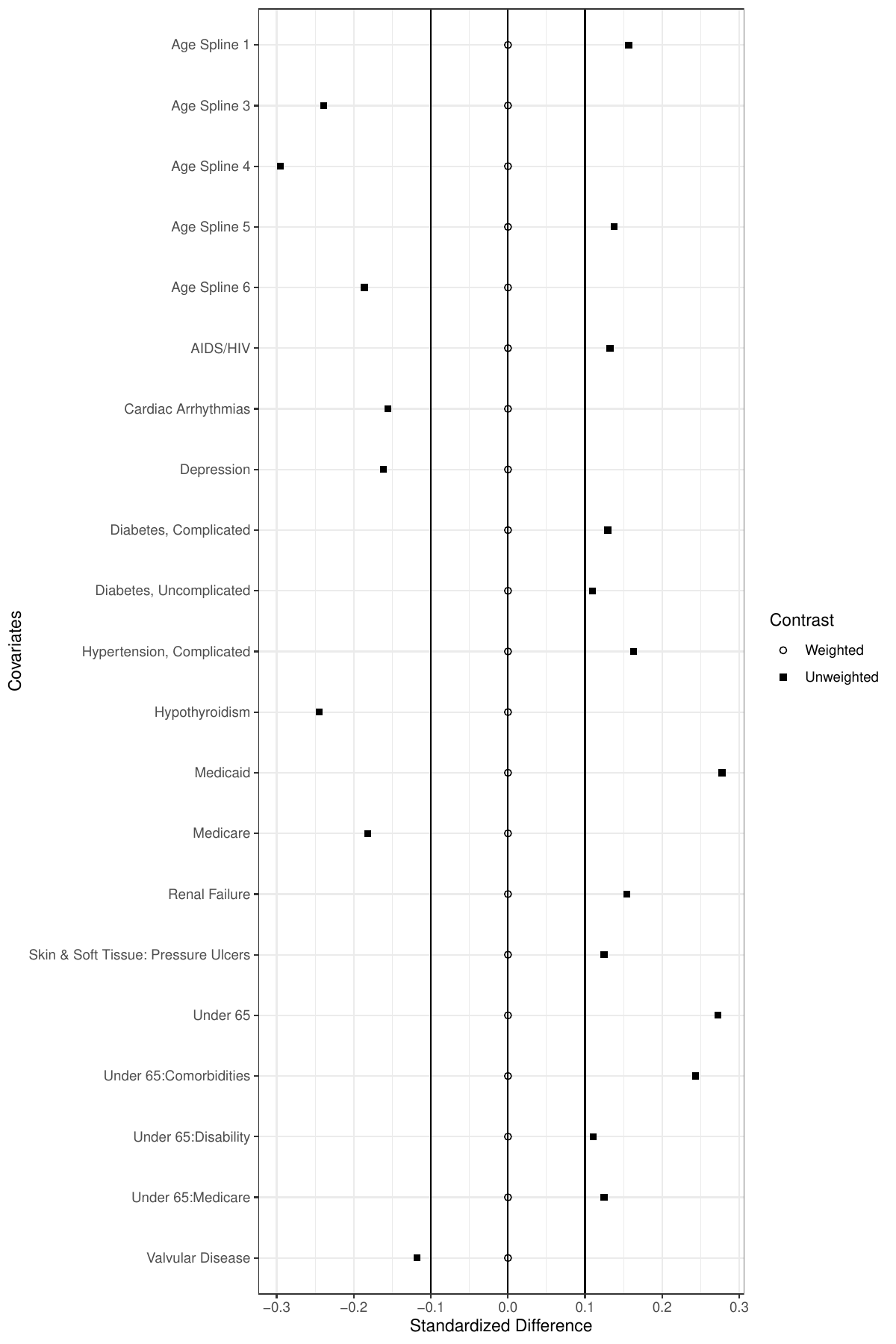}
     \caption{Standardized difference before and after weighting for the covariates with the largest baseline imbalances, not accounting for hospitals}
  \label{fig:bal.plot1}
\end{figure}

\begin{figure}[htbp]
  \centering
    \includegraphics[scale=0.6]{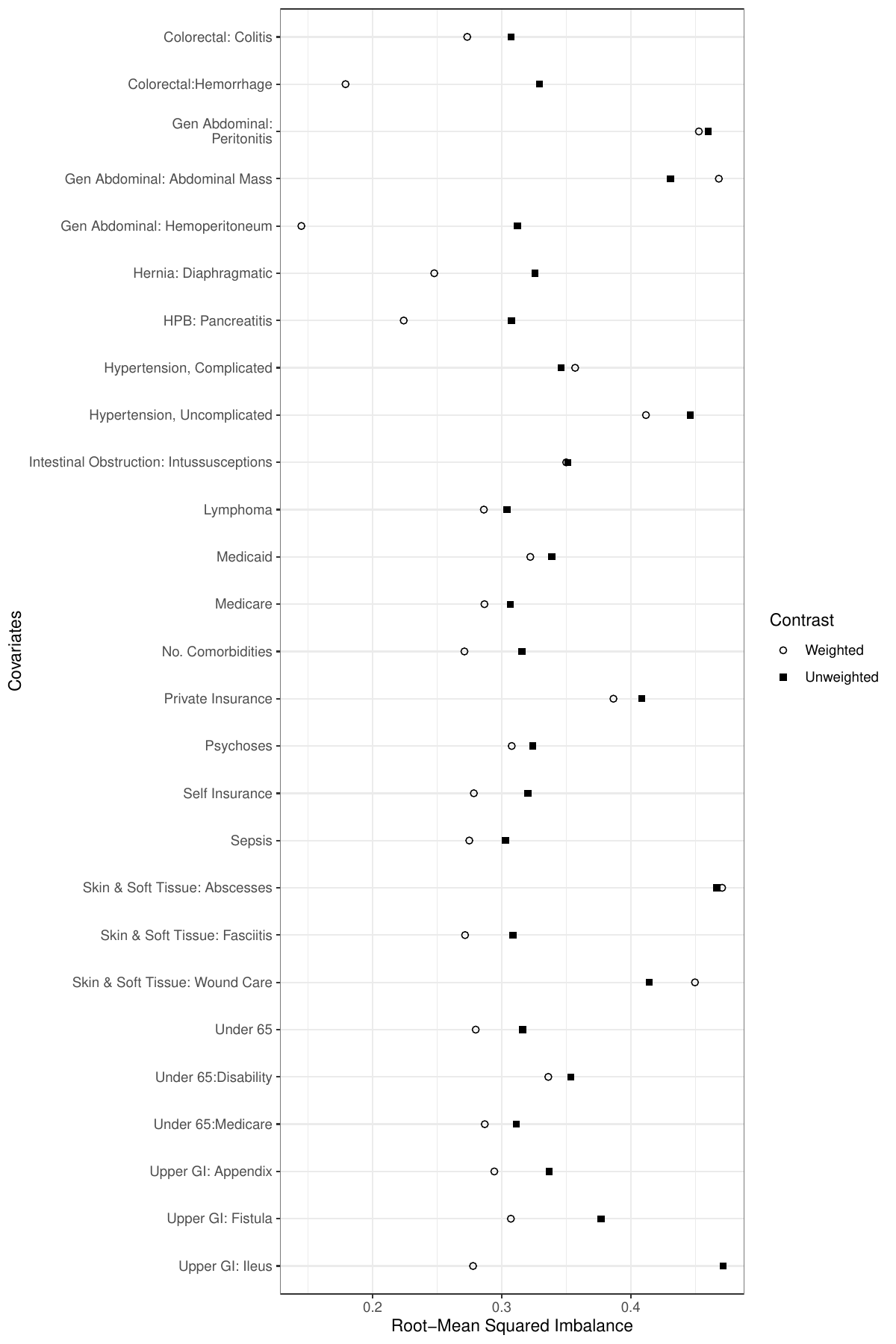}
     \caption{RMSI before and after weighting for the covariates with the largest baseline imbalances for surgery patients.}
  \label{fig:bal.plot2}
\end{figure}

\begin{figure}[htbp]
  \centering
    \includegraphics[scale=0.6]{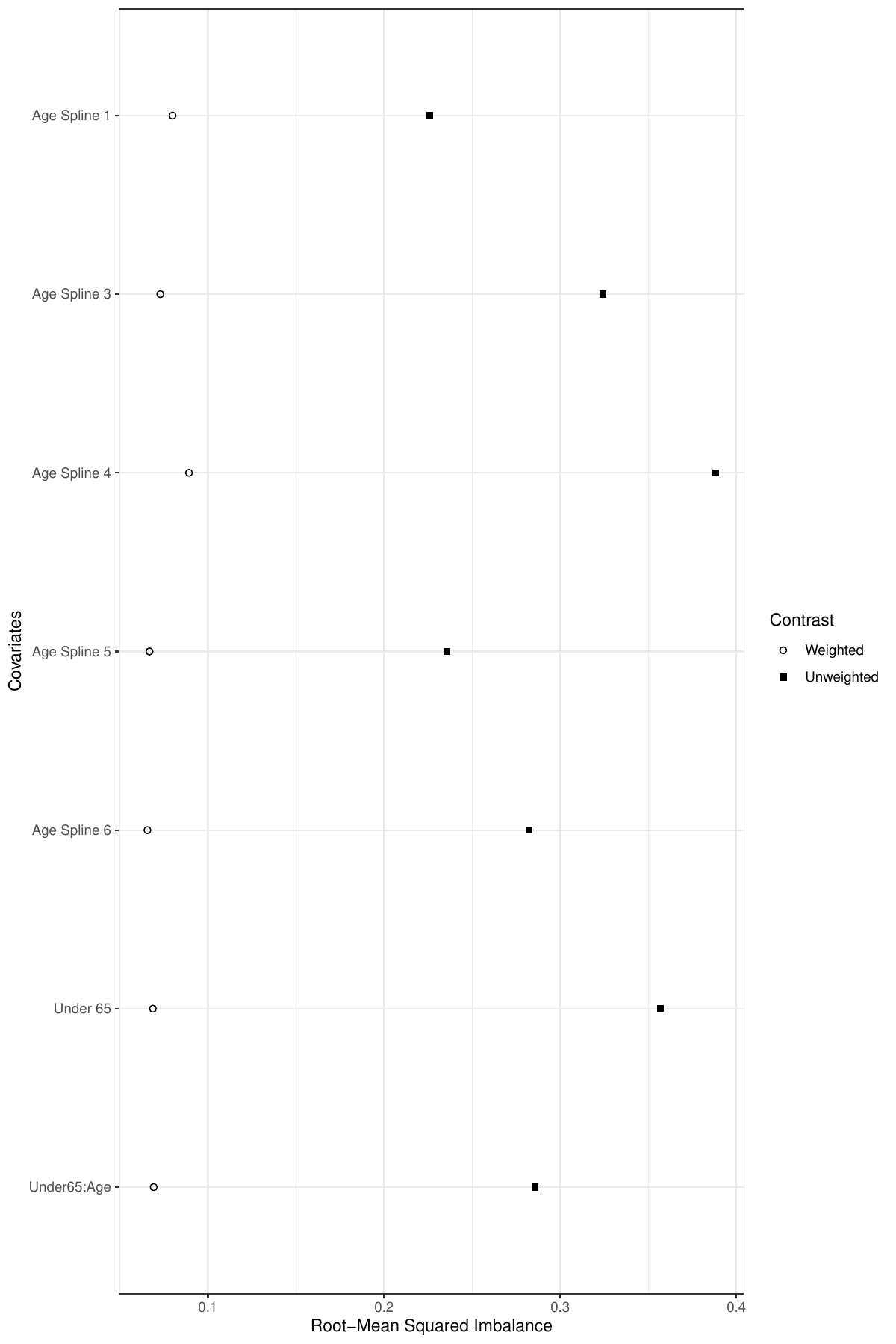}
     \caption{RMSI before and after weighting for the covariates with the largest baseline imbalances for outcome-allowable covariates.}
  \label{fig:bal.plot3}
\end{figure}

\subsection{Additional Disparity Estimates}
\label{sec:addl_estimates}
In the primary analysis, we used an indicator for an adverse event as the primary outcome. As we detailed in the main text, the adverse event indicator is 1 if the patient either died, had a complication, or prolonged length of stay. Here, we report the results using each of these measures as outcomes. Table~\ref{tab.egs1} contains the results for the full EGS population. Table~\ref{tab.egs2} contains the results for the subset of the EGS population that received surgery.

\begin{table}[htb]
\centering
\begin{threeparttable}
\caption{Estimated racial disparities for outcomes for the full EGS population, with different adjustment sets}
\label{tab.egs1}
\begin{tabular}{lcccc}
\toprule
& Unadjusted & $X^y$-adjusted & $X$-adjusted & $X$-adjusted disparity \\
 &            & disparity & disparity  & within hospital \\
\midrule
 \multirow{ 2}{3.5cm}{Death}  & 0.001 & 0.026 & 0.004 & -0.0001  \\  
 & [ 0 , 0.002 ] & [ 0.023 , 0.029 ] & [ 0.002 , 0.006 ] & [ -0.003 , 0.002 ]  \\  
\midrule
\multirow{ 2}{3.5cm}{Complication}  & 0.001 & 0.492 & 0.002 & 0.0001  \\  
& [ -0.002 , 0.003 ] & [ 0.487 , 0.498 ] & [ -0.001 , 0.006 ] & [ -0.007 , 0.006 ]  \\ 
\midrule
\multirow{ 2}{3.5cm}{Prolonged Length of Stay} & -0.036 & -0.023 & 0.013 & 0.007  \\  
& [ -0.037 , -0.034 ] & [ -0.027 , -0.019 ] & [ 0.011 , 0.016 ] & [ -0.007 , 0.006 ]  \\  
\bottomrule
\end{tabular}
\begin{tablenotes}[para]
{ \footnotesize Note: Numbers in brackets are 95\% confidence intervals.} 
\end{tablenotes}
\end{threeparttable}
\end{table}

\begin{table}[htb]
\centering
\begin{threeparttable}
\caption{Estimated Racial Disparities for Outcomes Following Surgery for an EGS Condition}
\label{tab.egs2}
\begin{tabular}{lcccc}
\toprule
& Unadjusted  & $X\text{-adjusted}$ & $X\text{-adjusted disparity}$ \\
 &             & disparity  & within hospital \\
\midrule
 \multirow{ 2}{3.5cm}{Death} & 0.003 & 0.004 & 0.0001  \\ 
 & [ 0.001 , 0.004 ] & [ 0.002 , 0.007 ] & [ -0.003 , 0.004 ]  \\  
\midrule
\multirow{ 2}{3.5cm}{Complication} & 0.0067 & 0.007 & 0.002  \\ 
& [ 0.003 , 0.01 ] & [ 0.002 , 0.012 ] & [ -0.007 , 0.011 ]  \\ 
\midrule
\multirow{ 2}{3.5cm}{Prolonged Length of Stay} & -0.026 & 0.017 & 0.009  \\  
& [ -0.028 , -0.023 ] & [ 0.014 , 0.02 ] & [ 0.003 , 0.014 ]  \\  
 \bottomrule
\end{tabular}
\begin{tablenotes}[para]
{ \footnotesize Note: Numbers in brackets are 95\% confidence intervals.} 
\end{tablenotes}
\end{threeparttable}
\end{table}

\clearpage
\bibliographystyle{chicago}
\bibliography{references}